\documentclass[11pt]{article}
\usepackage[margin=1in]{geometry}
\usepackage[utf8]{inputenc}
\usepackage[dvipsnames]{xcolor}
\usepackage{hyperref,url}
\hypersetup{
    colorlinks,
    linkcolor=black,
    citecolor=blue,
    urlcolor={blue!80!black}
}
\usepackage[T1]{fontenc}
\usepackage{tikz}
\usepackage{graphicx,subcaption}
\usepackage{float}
\usepackage{authblk}

\usepackage{csquotes}
\usepackage[normalem]{ulem}
\RequirePackage[inline,shortlabels]{enumitem}

\usepackage{amsmath,amsfonts,amsthm,amssymb,dsfont,bbm,bm}
\theoremstyle{plain}
\newtheorem{proposition}{Proposition}
\newtheorem{corollary}{Corollary}
\newtheorem{remark}{Remark}
\newtheorem{example}{Example}
\newtheorem{definition}{Definition}
\newtheorem{lemma}{Lemma}
\newtheorem{theorem}{Theorem}

\AtBeginDocument{\def\doi#1{\url{https://doi.org/#1}}}
\usepackage[backend=biber, style=authoryear, natbib=true]{biblatex}
\AtEveryBibitem{
  \clearfield{month}
  \clearfield{day}
}

\DeclareSourcemap{
  \maps[datatype=bibtex]{
    \map{
      \step[fieldsource=doi, final]
      \step[fieldset=url, null]
    }
  }
}

\usepackage{array,multirow}
\newcolumntype{P}[1]{>{\centering\arraybackslash}p{#1}}
\newcolumntype{M}[1]{>{\centering\arraybackslash}m{#1}}
\usepackage{float}

\usepackage{siunitx}

\usepackage{hyperref}
\hypersetup{
    colorlinks,
    linkcolor=blue,
    citecolor=blue,
    urlcolor={blue!80!black}
}

\usepackage{dsfont}
\usepackage{cases}

\newcommand{\calN}{\mathcal{N}}

\newcommand{\bbR}{\mathbb{R}}

\newcommand{\FB}{\mathrm{FB}}
\newcommand{\FAR}{\mathrm{FAR}}
\newcommand{\HR}{\mathrm{HR}}
\newcommand{\Pre}{\mathrm{Pre}}
\newcommand{\Rec}{\mathrm{Re}}

\newcommand{\onefun}{\mathds{1}}
\newcommand{\ind}[1]{\mathds{1}\{#1\}}

\usepackage{algorithm}
\usepackage{algpseudocode}

\title{Spatial Aggregation of ROC and Precision-Recall Curves}
\author[1,2,*]{Romain Pic}
\author[3]{Zhongwei Zhang}
\author[2]{Sebastian Engelke}
\author[1]{Johanna Ziegel}

\affil[1]{Seminar for Statistics, ETH Zurich, Zurich, 8092, Switzerland}
\affil[2]{GSEM-RISIS, University of Geneva, Geneva, 1205, Switzerland}
\affil[3]{Institute of Statistics, Karlsruhe Institute of Technology, Karlsruhe, 76185, Germany}
\affil[*]{Corresponding author: romain.pic@math.ethz.ch}

\begin{document}

\maketitle

\begin{abstract}
    Receiver Operating Characteristic (ROC) and Precision–Recall (PR) curves are widely used to assess the discrimination ability of forecasts for binary events, such as threshold exceedances or warnings of extreme events. In weather forecasting, forecasts are provided as spatial fields, yielding location-wise ROC and PR curves that are often aggregated to facilitate comparison. However, the effect of the aggregation strategy on performance assessment remains poorly understood.
    
    We investigate how different aggregation strategies for ROC and PR curves affect the assessment of discrimination ability. In particular, we identify conditions under which aggregation strategies satisfy two desirable properties for fair comparison: preservation of dominance between forecasts and preservation of concavity or achievability of the curves. We obtain sufficient conditions and propose two strategies satisfying them. They are compared with existing strategies from the literature, and we analyze their properties and highlight potential pitfalls that may lead to misleading interpretations. Based on these findings, we provide practical guidelines for the interpretation of aggregated ROC and PR curves. The proposed framework is illustrated with AI-based global weather forecasts, showing how different aggregation strategies can yield different rankings of competing forecasts.
\end{abstract}

\section{Introduction}

Binary events are central to weather forecasting, as many meteorological phenomena are naturally represented in this way: rain, frost, or tornadoes either occur or not. Continuous variables are also routinely reduced to binary events through the exceedance of a threshold, and such thresholds are typically imposed by the user rather than the forecaster: precipitation amounts triggering hydrological warnings, temperature thresholds for winter road maintenance, or wind speeds beyond which construction or transport operations are suspended \citep{MeteoSwiss2026DangerLevels}. In both cases, the resulting forecast feeds into the simplest decision-making setup, in which a single action is either taken or not, such as issuing a warning.

Two definitions of extreme events rely on binary events corresponding to the exceedance of a target threshold by a variable of interest. \textit{Absolute extreme events} are defined for a fixed threshold across locations, season, and time, and \textit{relative extreme events} use location-, season-, and time-specific thresholds often based on the climatology. Absolute extreme events can be of interest when the impact of the extreme event is not location-dependent. For example, MeteoSwiss determines heat warning levels based on daily mean temperatures exceeding 25°C and 27°C thresholds for three consecutive days \citep{MeteoSwiss2026}. This definition of heat warning levels assumes that the impact of heat events is homogeneous across Switzerland. On the other hand, relative extreme events can account for the potential impact of an event because resilience to an event often depends on its rarity, allowing for heterogeneity across the different locations, seasons, or times considered (e.g., rain warnings at MeteoSwiss; \cite{MeteoSwiss2026RainWarnings}).\\

\begin{figure}[ht]
    \centering
    \includegraphics[width=\linewidth]{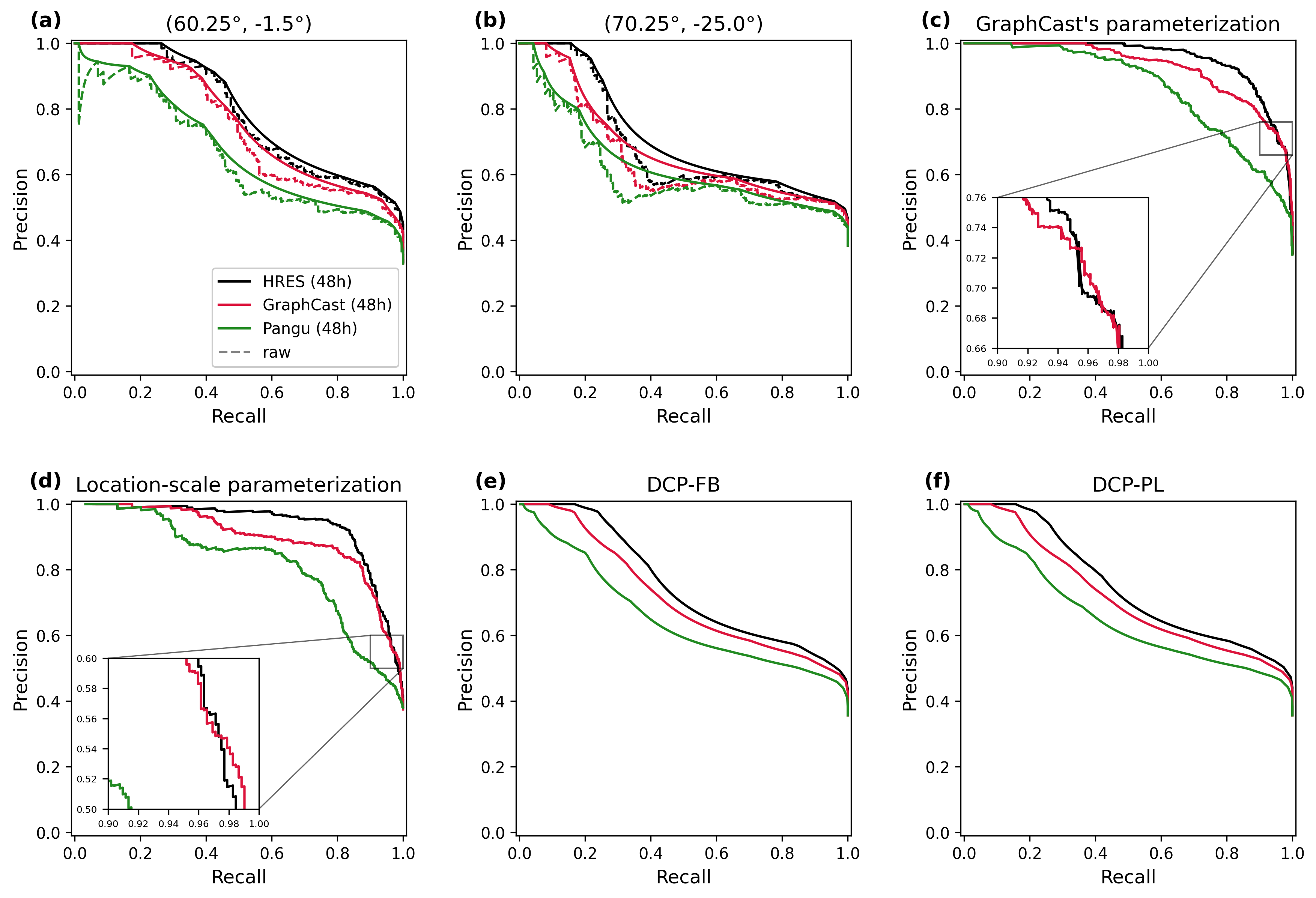}
    \caption{Location-wise raw (dashed lines) and achievable (solid) PR curves (a,b), and aggregated PR curves using GraphCast's parameterization (c), the location-scale parameterization (d), DCP-FB (e), and DCP-PL (f), for HRES (black), GraphCast (red), and Pangu (green) at a 2-day lead time over two grid points with latitude and longitude (60.25\textdegree, -1.5\textdegree) and (70.25\textdegree, -25.0\textdegree). The target thresholds correspond to the 75th percentiles of ERA5 for each location, time of day, and month of year over 1979-2019.}
    \label{fig:agg-PR-pair}
\end{figure}

Receiver operating characteristic (ROC) and precision-recall (PR) curves are popular tools for the verification of forecasts for binary events. ROC curves represent the discrimination ability of forecasts, i.e., their ability to distinguish between the occurrence and non-occurrence of the event \citep{Fawcett2006}. Moreover, both are standard methods for extreme event verification \citep[e.g.,][]{Lam2023,BenBouallegue2024}. Despite the equivalence between ROC and PR curves for comparing forecasts \citep{Davis2006}, they exhibit different graphical sensitivities that can lead to divergent visual interpretations when dealing with rare events \citep{Saito2015}. 

Recently, \citet{Lam2023} and \citet{Zhang2026} have aggregated PR curves over the globe to compare the performance of AI-based weather forecasts to that of physics-based forecasts at predicting extreme temperatures. The aggregation of ROC and PR curves aims to summarize the information contained in the individual location-wise ROC or PR curves. However, the effect of the aggregation strategy on performance assessment has not been studied, even though it can lead to questionable conclusions when done improperly. In Figure~\ref{fig:agg-PR-pair}, three forecasts are compared across two locations, and they have the same dominance relationship in terms of PR curves (a,b). Certain aggregation strategies do not preserve this dominance, leading to crossing aggregated curves (c,d). Such a crossing of PR curves implies that there is no single best forecast; rather, the best forecast depends on the end user's preference. Nonetheless, aggregation strategies that preserve dominance lead to aggregated curves with the same dominance relationships as those present location-wise (e,f). In this article, we propose dominance- and concavity-preserving aggregation strategies with FB parametrization (DCP-FB) and with parallel lines parametrization (DCP-PL).\\

The remainder of this article is organized as follows. Section~\ref{sec:background} introduces the necessary background on forecast verification for binary outcomes, and on ROC and PR curves in particular. Section~\ref{sec:spatial-agg} focuses on the aggregation of ROC and PR curves: Sections~\ref{subsec:agg-param-strat}~and~\ref{subsec:agg-interp-counts} set up a general framework and recall two aggregation strategies from the literature, while Sections~\ref{subsec:preservation-dom}~and~\ref{subsec:preservation-concavity-achievability} investigate the preservation of dominance, concavity, and achievability. Section~\ref{subsec:examples} shows by example why aggregating counts may fail to preserve these properties, and gives sufficient conditions under which aggregating interpolated counts preserves the properties. Section~\ref{sec:applications} illustrates the practical importance of these properties and compares the curves produced by the different strategies on the WeatherBench 2 dataset \citep{Rasp2024}. Section~\ref{sec:conclusion} summarizes the contributions, discusses future research directions, and provides recommendations. Proofs are postponed to the Appendix.

\section{Background}\label{sec:background}

\subsection{Basic performance measures}

We consider a forecast $X\in\bbR$ and an observation $Y\in\bbR$. We let $(x_i, y_i), i=1,\dots,n$, be $n$ realizations of the forecast-observation pair $(X,Y)$. An event, potentially extreme, is predicted if the forecast exceeds a certain threshold $t_X$, called the \textit{decision threshold}. Similarly, an event is observed if the observation exceeds a \textit{target threshold} $t_Y$. We will consider the target threshold $t_Y$ to be fixed, so the observations $y_i$ are reduced to binary observations. For a given decision threshold $t_X$, the comparison between forecasts and observations can be summarized by the counts of the four possible cases presented in the \textit{contingency table} (Table~\ref{table:contingency}). The coefficients of the contingency table are the counts of \textit{hits}, $a$, \textit{false alarms}, $b$, \textit{misses}, $c$, and \textit{correct rejections}, $d$. The coefficients sum to the sample size $n=a(t_X)+b(t_X)+c(t_X)+d(t_X)$, for all $t_X\in\bar{\bbR}$. In the machine learning literature, the contingency table is called the \textit{confusion matrix}, and its entries are called the true positives, $a$, the false positives, $b$, the false negatives, $c$, and the true negatives, $d$.

\begin{table}
    \caption{Contingency table.}
    \centering
    \renewcommand{\arraystretch}{2}
    \begin{tabular}{|M{.65cm} M{1.25cm}|M{6.25cm}|M{6.25cm}|}
        \cline{3-4}
        \multicolumn{2}{c|}{}&\multicolumn{2}{c|}{Observation}\\
        \multicolumn{2}{c|}{}& $Y> t_Y$ & $Y\leq t_Y$ \\
        \hline
        \multirow{2}{*}{\rotatebox[origin=c]{90}{\parbox[c]{1.75cm}{\centering Forecast}}} & \small{$X> t_X$} & $a(t_X)= \sum_{i=1}^n \ind{x_i>t_X}\ind{y_i> t_Y}$ & $b(t_X)= \sum_{i=1}^n \ind{x_i>t_X}\ind{y_i\leq t_Y}$ \\
        \cline{2-4}
        & \small{$X\leq  t_X$} & $c(t_X)= \sum_{i=1}^n \ind{x_i\leq t_X}\ind{y_i> t_Y}$ & $d(t_X)= \sum_{i=1}^n \ind{x_i\leq t_X}\ind{y_i\leq t_Y}$\\
        \hline
    \end{tabular}
    \label{table:contingency}
\end{table}

The \textit{false alarm rate} (FAR), the \textit{hit rate} (HR) or \textit{recall}, and the \textit{precision} for the decision threshold $t_X$ are given by
\begin{align}
    \FAR(t_X) &= \frac{\sum_{i=1}^n \ind{x_i>t_X}\ind{y_i\leq t_Y}}{\sum_{i=1}^n \ind{y_i\leq t_Y}} = \frac{b(t_X)}{b(t_X)+d(t_X)},\label{eq:FAR}\\
    \HR(t_X) = \Rec(t_X) &= \frac{\sum_{i=1}^n \ind{x_i>t_X}\ind{y_i>t_Y}}{\sum_{i=1}^n \ind{y_i> t_Y}} = \frac{a(t_X)}{a(t_X)+c(t_X)},\label{eq:HR-Re}\\
    \Pre(t_X) &= \frac{\sum_{i=1}^n \ind{x_i>t_X}\ind{y_i>t_Y}}{\sum_{i=1}^n \ind{x_i> t_X}} = \frac{a(t_X)}{a(t_X)+b(t_X)}.\label{eq:Pre}
\end{align}
The \textit{base rate} is the proportion of observed occurrences in the sample, i.e.,
\begin{equation*}
    p=\frac{\sum_{i=1}^n \onefun\{y_i>t_Y\}}{n} = \frac{a(t_X)+c(t_X)}{n},
\end{equation*}
and is a characteristic of the observations only and thus independent of $t_X$. The \textit{frequency bias} (FB) is the ratio between the number of forecast occurrences and of observed ones, i.e.,
\begin{equation*}
    \FB(t_X) = \frac{\sum_{i=1}^n\onefun\{x_i>t_X\}}{\sum_{i=1}^n\onefun\{y_i>t_Y\}}  = \frac{a(t_X)+b(t_X)}{a(t_X)+c(t_X)}.
\end{equation*}

Given a fixed sample size $n$, performance measures (such as recall or precision) can be expressed in terms of the base rate, the FAR, and the HR \citep{Hogan2011}. Indeed, recall and precision can be expressed as
\begin{gather}
    \Rec = \HR;\label{eq:Re-pFH}\\
    \Pre = \frac{p\cdot\HR}{p\cdot\HR+(1-p)\cdot\FAR}.\label{eq:Pre-pFH}
\end{gather}

\subsection{Receiver Operating Characteristics (ROC)}

We recall the concept of receiver operating characteristic (ROC) curves and present related notions. We follow \citet{GneitingVogel2022} and let the \textit{raw ROC diagnostic} be the union of the points of the form $(\FAR(t_X), \HR(t_X))$, where $t_X\in\bar{\bbR}:=[-\infty, \infty]$. The \textit{(raw) ROC curve} is obtained by linearly interpolating points of the raw ROC diagnostic. The linear interpolation is justified by the fact that any point between two points in the ROC space can be achieved by randomly changing the forecast values of one point of the curve into those of the other point \citep{Hogan2011}; details are given in Remark \ref{remark:random-mixture}.

\begin{remark}\label{remark:random-mixture}
    Let $t_X,t_X'\in\bar{\bbR}$ be two distinct decision thresholds. Let $t_Y$ be a fixed target threshold. For a given sample of forecast-observation pairs $(x_1,y_1),\dots, (x_n,y_n)\in\bbR^2$, $t_X$ and $t_X'$ correspond to two contingency tables and thus two points in ROC space. Without any additional information, new threshold exceedances can be obtained by randomly drawing a decision threshold between $t_X$ and $t_X'$, i.e., a proportion $\alpha$ of the forecast exceedances are computed with respect to $t_X$ and a proportion $1-\alpha$ are computed with respect to $t_X'$. If all the forecast realizations are compared with $t_X$, $(a(t_X),b(t_X),c(t_X),d(t_X))$ is obtained, and if all the forecast realizations are compared with $t_X'$, $(a(t_X'),b(t_X'),c(t_X'),d(t_X'))$ is obtained. For any other proportion of switched forecast realizations, the expected contingency table with respect to the random draw of the decision threshold corresponds to the component-wise interpolated contingency table,
    \begin{equation*}\label{eq:contingency-interpolation}
        \begin{cases}
            a(t_X'')=\alpha a(t_X)+ (1-\alpha) a(t_X')\\
            b(t_X'')=\alpha b(t_X)+ (1-\alpha) b(t_X')\\
            c(t_X'')=\alpha c(t_X)+ (1-\alpha) c(t_X')\\
            d(t_X'')=\alpha d(t_X)+ (1-\alpha) d(t_X')\\
        \end{cases},
    \end{equation*}
    where $t_X''=\alpha t_X+(1-\alpha)t_X'$, with $\alpha\in[0,1]$ being the proportion of forecast samples compared with $t_X$. In ROC space, the corresponding point is on the segment between $(\FAR(t_X),\HR(t_X))$ and $(\FAR(t_X'),\HR(t_X'))$ since only the numerators in \eqref{eq:FAR} and \eqref{eq:HR-Re} are affected. The resulting FAR and HR are expressed as
    \begin{align*}
        \FAR(t_X'') &= \alpha\FAR(t_X)+(1-\alpha)\FAR(t_X');\\ 
        \HR(t_X'') &= \alpha\HR(t_X)+(1-\alpha)\HR(t_X').
    \end{align*}
\end{remark}

Without loss of generality, we identify any ROC curve with a function $R \in \mathcal{R}$, where $\mathcal{R}$ is the class of all functions $R:[0,1]\to[0,1]$ that are right-continuous, non-decreasing, and satisfy the boundary conditions $R(0) = 0$ and $R(1) = 1$. Any ROC curve falls into this class \citep[Theorem~4]{GneitingVogel2022}.

The ROC curve is given by a family of contingency tables indexed by $t\in\bar{\bbR}$. We consider a family of contingency tables to be composed of contingency tables taking real positive values,
\begin{equation*}
\left\{
    \begin{bmatrix}
        a(t) & b(t) \\
        c(t) & d(t)
    \end{bmatrix}, t\in\bar{\bbR}\right\},
\end{equation*}
such that, for all $t\in\bar{\bbR}$, $a(t)+b(t)+c(t)+d(t)$ and $a(t)+c(t)$ are constant and $a(t)$ and $b(t)$ are either both increasing or both decreasing with $t$ taking values from $0$ to $a(t)+c(t)$ and $b(t)+d(t)$, respectively (both included). Note that the range of the parameter $t$ and the choice between increasing or decreasing $a(t)$ and $b(t)$ are conventions. Such a family of contingency tables can be used to build a ROC curve. 
This parameterization of a family of contingency tables provides flexibility to the indexing of ROC curves beyond the ones based on the decision threshold $t_X$ and the $\FAR$.\\

Despite belonging to a broader class of functions, ROC curves ought to be concave. As ROC curves aim to assess potential predictive performance, they should assess the predictive performance contained in the forecast without additional information. The best ROC curve that can be obtained by allowing for random mixture (see Remark~\ref{remark:random-mixture}) is the concave hull of the raw ROC diagnostic, which is (as its name indicates) concave. Formally, a ROC curve $R\in\mathcal{R}$ is \textit{concave} if for any three different points $(x_1,y_1),(x_2,y_2),(x_3,y_3)\in R$ such that $x_1\leq x_2\leq x_3$ and $y_1\leq y_2\leq y_3$, it holds that $(y_2-y_1)/(x_2-x_1) \geq (y_3-y_2)/(x_3-x_2)$. For a fixed target threshold $t_Y$, the concavity of ROC curves is essentially equivalent to the conditional event probability (CEP) of the forecast-observation pair $(X,Y)\in\bbR^2$ associated with the ROC curve,
\begin{equation*}
    \mathrm{CEP}(t):= \Pr(Y>t_Y\mid X=t),
\end{equation*}
being non-decreasing with $t$ \citep[Theorems~2~and~3]{GneitingVogel2022}. This equivalence provides an alternative motivation for ROC curves to be concave since a non-decreasing CEP justifies thresholding the forecast \citep{Dimitriadis2024}. Indeed, it implies that the event ``$Y>t_Y$'' has a higher conditional probability as the forecast value $X$ increases and motivates reducing a forecast $X$ to a hard classifier based on a decision threshold $t$.\\

In practice, the postprocessed forecast values, $\hat{x}_i$, leading to the concave hull of a raw ROC diagnostic (i.e., the best achievable ROC curve) can be obtained using the Pool-Adjacent-Violators (PAV) algorithm (\cite{Gneiting2023}; Algorithm~\ref{alg:PAV} in Appendix~\ref{appendix:PAV}).

The PAV-transformed forecast $\hat{X}$ satisfies 
\begin{equation}\label{eq:calibration}
    \mathrm{CEP}(t) := \widehat{\Pr}(Y>t_Y\mid\hat{X}=t) = t,
\end{equation}
for all $t\in\{\hat{x}_1,\ldots,\hat{x}_n\}\subset[0,1]$, where $\widehat{\Pr}$ is the empirical distribution of $(\hat{x}_1,y_1),\dots,(\hat{x}_n,y_n)$ and $(\hat{X},Y)$ is a random draw from $\widehat{\Pr}$. Thus, it is calibrated as Eq.~\eqref{eq:calibration} corresponds to the unified notion of calibration for probabilistic forecasts of binary outcomes \citep{Gneiting2013}. PAV-transformed forecasts can be thought of as probabilistic forecasts for binary outcomes as they take values in $[0,1]$ (see Remark~\ref{remark:prob-forecast}). Moreover, their associated ROC curve is concave since the identity function is non-decreasing. We illustrate how the PAV-transformed forecast modifies the ROC curve in Figure~\ref{fig:PAV-ROC_PR}. The resulting curve is the concave hull of the raw ROC curve. Note that the PAV transformation does not affect concave ROC curves \citep[Theorem~1]{Fawcett2007}.\\

\begin{figure}
    \centering
    \includegraphics[width=0.8\linewidth]{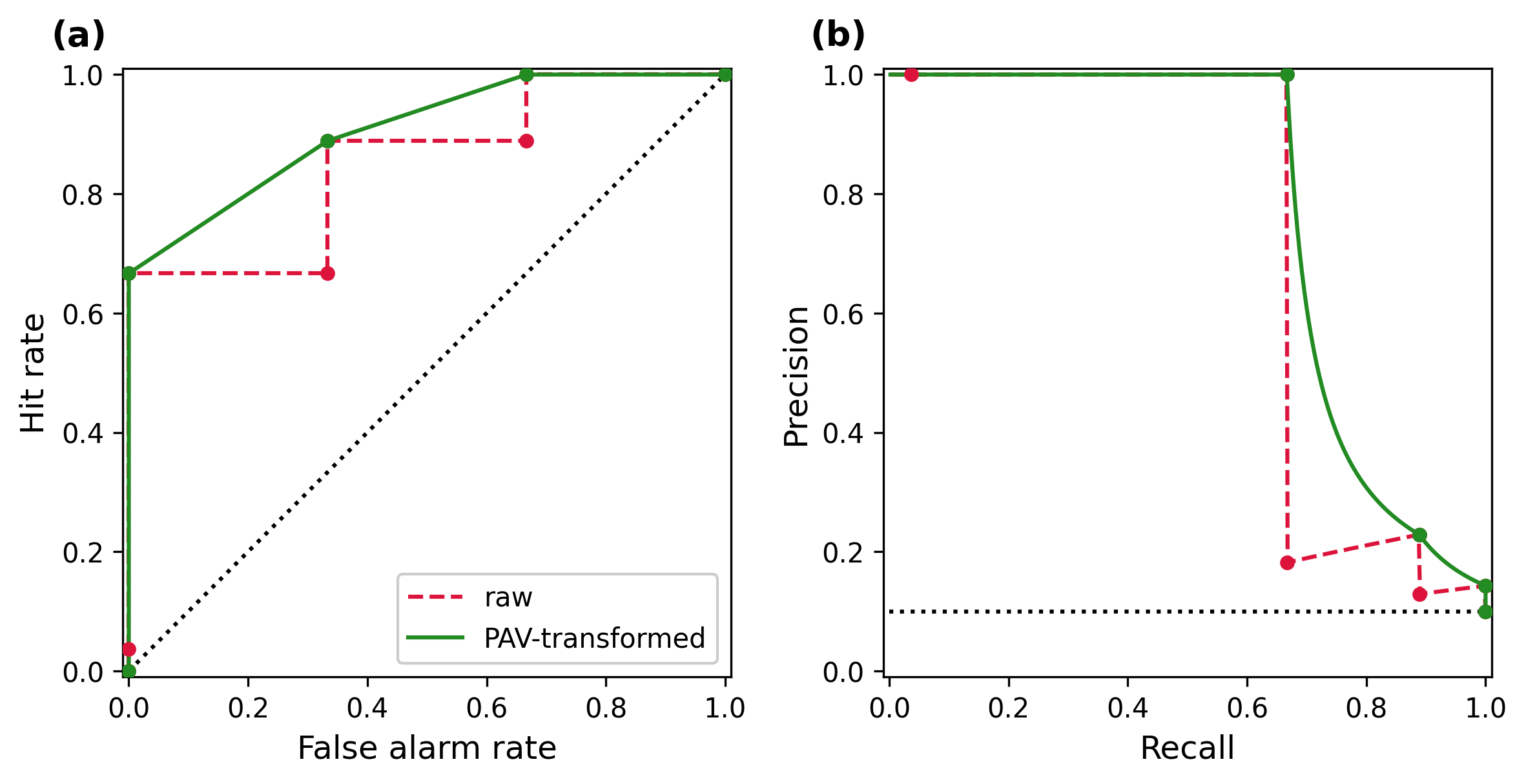}%
    \caption{Illustration of the effect of the PAV algorithm on ROC (a) and PR curves (b). The dotted black lines represent the no-skill line in ROC and PR space, respectively.}
    \label{fig:PAV-ROC_PR}
\end{figure}

A ROC curve assesses the potential predictive ability of a forecast. If curves do not cross, they can be used to rank competing forecasts. In this case, we speak of \textit{dominance}.

\begin{definition}
    Let $R_A,R_B\in\mathcal{R}$ be ROC curves. We say that $R_A$ \textit{dominates} $R_B$ if, for all $x\in[0,1]$,
    $R_A(x)\geq R_B(x)$, and there is at least one $x^\ast\in[0,1]$ such that the inequality is strict.
\end{definition}
This dominance between ROC curves can be used to induce a notion of dominance between forecasts. Let $X_A$ and $X_B$ be forecasts and $Y$ an observation. We denote $\hat{R}_A\in\mathcal{R}$ and $\hat{R}_B\in\mathcal{R}$ the concave ROC curves associated with $X_A$ and $X_B$, respectively. Forecast $X_A$ \textit{dominates} forecast $X_B$ \textit{in terms of ROC curves} if $\hat{R}_A$ dominates $\hat{R}_B$.

Dominance can be defined for any two distinct non-crossing ROC curves. However, it is possible that there is dominance between the concave ROC curve of one forecast and another, without having dominance between their raw ROC curves. To avoid confusion, we define the dominance between two competing forecasts in terms of ROC curves as relying on their concave ROC curves (i.e., the concave hull of their raw ROC curve).\\

The above definitions are deliberately related to the representation of ROC curves as functions of the FAR. However, these definitions are not restrictive and can be given equivalently in other parameterizations (e.g., in terms of FB, as in \cite{Dimitriadis2024}). Appendix~\ref{appendix:dominance-equiv} shows the equivalence with another definition of dominance between ROC curves not as functions of FAR.

\begin{remark}\label{remark:prob-forecast}
    \citet{Dimitriadis2024} showed the strong relationship between ROC curves and Murphy curves of probabilistic forecasts for binary outcomes. A \textit{Murphy curve} plots the value of elementary scores, which are scoring functions consistent for the mean and parameterized by a decision threshold \citep{Ehm2016}. When comparing two calibrated forecasts, dominance in terms of ROC curves and in terms of Murphy curves are equivalent, and, moreover, the number of crossings of the ROC curves equals the number of crossings of the Murphy curves \citep{Dimitriadis2024}.
\end{remark}

\subsection{Precision-Recall (PR)}

We recall the concept of precision-recall (PR) curves. Let the \textit{raw PR diagnostic} be the union of the points of the form $(\Rec(t_X), \Pre(t_X))$, where $t_X\in\bar{\bbR}$. 

In contrast to its ROC counterpart, the \textit{(raw) PR curve} cannot be obtained by linearly interpolating points from the raw PR diagnostic, as such points are not necessarily achievable \citep{Davis2006}. The interpolation between two points $(\Rec(t_X),\Pre(t_X))$ and $(\Rec(t_X'),\Pre(t_X'))$ in PR space is obtained by investigating their random mixture (see Remark~\ref{remark:random-mixture}). When $\Rec(t_X)\leq\Rec(t_X')$ and $\Pre(t_X)\geq\Pre(t_X')$, the interpolation should be done in the following way
\begin{equation}\label{eq:interp-PR}
    \begin{cases}
        \Pre(t_X'')=\cfrac{p\cdot\Rec(t_X'')}{p\cdot\Rec(t_X'')+(1-p)\cdot\left(\FAR(t_X)+\frac{\FAR(t_X')-\FAR(t_X)}{\HR(t_X')-\HR(t_X)}\cdot(\Rec(t_X'')-\HR(t_X))\right)} & \text{if } \Rec(t_X')\neq \Rec(t_X)\\
        \Pre(t_X'')=\cfrac{p\cdot\Rec(t_X)}{p\cdot\Rec(t_X)+(1-p)\cdot\FAR(t_X)} & \text{otherwise}
    \end{cases},
\end{equation}
with $\Rec(t_X'')=\alpha\Rec(t_X)+(1-\alpha)\Rec(t_X')$ and $t_X''=\alpha t_X+(1-\alpha)t_X'$, where $(\FAR(x),\HR(x))$ is the point in ROC space associated with $(\Rec(x),\Pre(x))$ using \eqref{eq:Re-pFH}~and~\eqref{eq:Pre-pFH}, with $x\in\{t_X,t_X'\}$. However, when $\Pre(t_X)<\Pre(t_X')$ or $\Rec(t_X)=\Rec(t_X')$ (as for the vertical and increasing parts of the raw PR curve in Figure~\ref{fig:PAV-ROC_PR}), the interpolation reduces to linear interpolation.\\

Given a fixed sample size $n$, there is a one-to-one correspondence between the ROC curve and the PR curve associated with the same forecast-observation pair, except for the point with $\Rec=\HR=0$ \citep[Theorem~3.1]{Davis2006}. Despite this correspondence, the class of PR curves and its desirable properties have a less elegant representation than that of its ROC counterpart. PR curves can be represented as functions $P:[0,1]\to[0,1]$ that are left-continuous and satisfy $P(1)=p$, where $p$ is the base rate; we denote by $\mathcal{P}$ the class of PR curves. They are not necessarily monotone (see Figure~\ref{fig:PAV-ROC_PR}), and the limit of precision as the recall approaches $0$ does not satisfy strict boundary conditions. As for ROC curves, a PR curve can be obtained from a family of contingency tables, allowing for parameterizations beyond those based on the decision threshold $t_X$ and recall.\\

Concavity is not a desirable property for PR curves. This is not related to the interpolation \eqref{eq:interp-PR} as PR diagnostics of PAV-transformed forecasts are not concave (see Figure~\ref{fig:PAV-ROC_PR}). The counterpart of the concave hull of a ROC curve is the \textit{achievable PR curve} \citep{Davis2006}. We consider a forecast $X$ and observations $Y$. We denote by $R$ and $P$ the raw ROC curve and the raw PR curve, respectively, associated with it. The achievable PR curve of $X$, $\hat{P}$, is the PR curve associated with the concave hull of $R$. Equivalently, it is the PR curve associated with the PAV-transformed forecast values. The achievable PR curve should be used when assessing discrimination ability.\\

Analogously to ROC curves, PR curves can be compared and used to derive a notion of \textit{dominance} between forecasts.

\begin{definition}
     Let $P_A, P_B \in \mathcal{P}$ be PR curves. We say that $P_A$ \textit{dominates} $P_B$ if, for all $x\in[0,1]$, $P_A(x)\geq P_B(x)$, and there is at least one $x^\ast\in[0,1]$ such that the inequality is strict.
\end{definition}
This dominance between PR curves induces a notion of dominance between forecasts. Let $X_A$ and $X_B$ be forecasts and $Y$ an observation. We denote by $\hat{P}_A\in\mathcal{P}$ and $\hat{P}_B\in\mathcal{P}$ the achievable PR curves associated with $X_A$ and $X_B$, respectively. Forecast $X_A$ \textit{dominates} forecast $X_B$ \textit{in terms of PR curves} if $\hat{P}_A$ dominates $\hat{P}_B$. Dominance can be defined for any two distinct non-crossing PR curves. However, dominance in terms of PR curves between forecasts relies on their achievable PR curves.\\

Dominance in terms of ROC curves and dominance in terms of PR curves are equivalent. The following theorem is a direct corollary of \citet[Theorem~3.2]{Davis2006}.
\begin{theorem}\label{thm:equiv-dominance}
    Let $X_A$ and $X_B$ be two forecasts and $Y$ be an observation. $X_A$ dominates $X_B$ in terms of ROC curves if and only if $X_A$ dominates $X_B$ in terms of PR curves.
\end{theorem}
\citet{Krueger2021} define \textit{forecast dominance} as having a better expected score for all scoring functions consistent with the mean functional, and provide generic conditions for it. In the case of calibrated forecasts for binary outcomes \eqref{eq:calibration}, this definition is equivalent to dominance in terms of Murphy curves and, thus, to dominance in terms of ROC or PR curves.\\

\begin{remark}
    The dominance between ROC curves yields the only objective ranking of ROC curves. Any other ordering is subjective as it assumes an end-user. In particular, any total ordering (i.e., ensuring that any two ROC curves are comparable) is subjective and end-user dependent. For example, PR curves have been summarized using a particular point of the PR curve such as the point of highest Symmetric Extreme Dependence Index \citep[SEDI;][]{Ferro2011}; see \citet{Lam2023} for application in AI-based weather forecasting. The area under the ROC curve (AUC) is the most popular univariate summary based on the entire curve, but it also implies a subjective ranking; see, for example, \citet{Atger2001} for further discussion.
\end{remark}

\section{Spatial aggregation of ROC and PR curves}\label{sec:spatial-agg}

\subsection{Parameterization strategies}\label{subsec:agg-param-strat}

We consider ROC and PR curves over $d$ spatial locations. Within this setting, let these curves be associated with families of contingency tables with coefficients denoted $(a_s(t), b_s(t),c_s(t),d_s(t))$ at location $s=1,\dots,d$, and parameterized by $t\in\bar{\bbR}$. 

To aggregate multiple curves across locations, we need a unified parametrization. Using the same parameter value for all locations is not sensible, since the forecast distributions may differ across locations. Similarly, the observed event distribution might also vary across locations. When aggregating ROC or PR curves, one must jointly define the parameters across locations $(t_{1}, \dots, t_{d})\in\bar{\bbR}^d$, and, moreover, how they jointly vary. This is specified by a \textit{parameterization strategy}.
\begin{definition}
    A \textit{parameterization strategy} is a mapping $\gamma: [0,1] \to \bar{\bbR}^d$ that is non-decreasing in each component, and satisfies the boundary conditions $\gamma(0)=(-\infty,\dots,-\infty)$ and $\gamma(1)=(\infty,\dots,\infty)$.
\end{definition} 
Without loss of generality, the domain for the parameters can be any compact and convex subset of $\bar{\bbR}^d$ and, in that case, the boundary conditions should match the parameters' range boundaries. For a common parameter $u\in[0,1]$, the parameter value at location $s=1,\dots,d$ is $t_s(u)=\gamma_s(u)$, where $\gamma_s(u)$ is the $s$-th component of $\gamma(u)$. We provide two examples of parameterization strategies from the literature, both of which are based on decision thresholds. 
\begin{example}{(GraphCast's parameterization strategy)}\label{example:GC}
    To aggregate PR curves in the context of AI-based weather forecasting, \citet{Lam2023} rescale forecasts by
    \begin{equation*}
        X'_s = \mathrm{med}_s + g(u) \cdot (X_s-\mathrm{med}_s),
    \end{equation*}
    where $X_s$ is the original forecast at location $s$, $g(u)\in[0,+\infty)$ is a gain parameter, and $\mathrm{med}_s$ is the climatological median at location $s$ for the observation. Varying the gain parameter $g(u)$ as a function of $u\in[0,1]$ for a fixed thresholding on the rescaled forecast (i.e., $X'_s>t_{Y,s}$) is equivalent to using the parameterization strategy
    \begin{equation}\label{eq:thresholding-strategy-GC}
            \gamma_s(u) = \begin{cases}
                -\infty & \text{if } u=0,\\
                \mathrm{med}_s + \frac{1}{g(u)} (t_{Y,s}-\mathrm{med}_s) & \text{if } u\in(0,1),\\
                \infty & \text{if } u=1,
            \end{cases}
    \end{equation}
    where $g:(0,1)\to(0,\infty)$ is a strictly monotone bijection. The target threshold $t_{Y,s}$ is taken to be the 98th climatological percentile at location $s$ for the observation. Hereafter, we refer to this parameterization strategy as \textit{GraphCast's parameterization strategy}.
\end{example}

\begin{example}{(Location-scale parameterization strategy)}\label{example:Zhang}
    \citet{Zhang2026} propose a different parameterization strategy motivated by assuming probabilistic forecasts from a location-scale family. It corresponds to the following rescaling
    \begin{equation*}
        X_s' = X_s + g(u) \cdot \sigma_{X,s},
    \end{equation*}
    where $X_s$ is the original forecast at location $s$, $g(u)\in\bbR$ is a gain parameter, and $\sigma_{X,s}$ is the climatological standard deviation of the forecast at location $s$. Varying the gain parameter for a fixed threshold on the rescaled forecast is equivalent to varying the threshold with the parameterization strategy
    \begin{equation}\label{eq:thresholding-strategy-ZZ}
            \gamma_s(u) = t_{Y,s} - g(u)\cdot \sigma_{X,s},
    \end{equation}
    where $g:[0,1]\to\bar{\bbR}$ is a strictly monotone bijection. If $X_s$ is from a location and scale family with location parameter $\mu_{X,s}$ and scale parameter $\sigma_{X,s}$ (i.e., $X_s=\mu_{X,s}+\sigma_{X,s} X$, where $X$ follows the standard distribution of the family), then $X_s>\gamma_s(u)$ becomes
    \begin{equation*}
        X>\frac{t_{Y,s}-\mu_{X,s}}{\sigma_{X,s}}+g(u),
    \end{equation*}
    where $X$ follows the standard distribution of the location-scale family, and, if the target thresholds are taken to be such that $t_{Y,s}=\mu_{X,s}+\sigma_{X,s}\cdot t$ (with $t\in\bbR$), this corresponds to taking the decision thresholds to be quantiles of the forecast distribution at all locations. As \citet{Zhang2026} focuses on record-breaking events, the target thresholds are set to the calendar-month record over the training period at each location. Hereafter, we refer to this parameterization strategy as the \textit{location-scale parameterization strategy}.
\end{example}

At each location $s=1,\dots,d$, the FAR, the HR or recall, and the precision associated with the forecast are denoted with a subscript and defined using the coefficients of the contingency table corresponding to the parameter $\gamma_s(u)$.

In addition to a parameterization strategy, one needs to define an \textit{aggregation strategy} specifying how location-wise quantities are aggregated.

\subsection{Aggregating interpolated counts}\label{subsec:agg-interp-counts}

A common way to aggregate is at the level of the contingency tables.  At each location, the family of contingency tables is continuous. An aggregated contingency table is obtained by summing location-wise components. Then, the desired quantities are computed to obtain an aggregated ROC or PR curve. For a given parameterization strategy $\gamma$ and $u\in[0,1]$, the aggregated FAR, HR or recall, and precision are defined as 
\begin{align}
    \FAR(\gamma(u)) &= \frac{\sum_{s=1}^d b_s(\gamma_s(u))}{\sum_{s=1}^d b_s(\gamma_s(u))+d_s(\gamma_s(u))};\label{eq:FAR-interp-counts}\\
    \HR(\gamma(u)) = \Rec(\gamma(u)) &= \frac{\sum_{s=1}^d a_s(\gamma_s(u))}{\sum_{s=1}^d a_s(\gamma_s(u))+c_s(\gamma_s(u))};\label{eq:HR_Re-interp-counts}\\
    \Pre(\gamma(u)) &= \frac{\sum_{s=1}^d a_s(\gamma_s(u))}{\sum_{s=1}^d a_s(\gamma_s(u))+b_s(\gamma_s(u))}.\label{eq:Pre-interp-counts}
\end{align}

\citet{Lam2023} and \citet{Zhang2026} use an aggregation strategy based on aggregating counts along with GraphCast's parameterization strategy \eqref{eq:thresholding-strategy-GC} and the location-scale parameterization strategy \eqref{eq:thresholding-strategy-ZZ}, respectively, to compute spatially aggregated PR curves. If, instead of aggregating families of continuous contingency tables, we aggregate families of contingency tables that take only positive integer values, this corresponds to aggregating counts.\\

Aggregating counts and aggregating interpolated counts rely on the assumption that the events of interest have the same meaning across locations. Since the counts are not weighted when aggregated, a count should have the same importance in terms of impact or rarity despite its location. This implies a uniformity of the impact over the spatial domain considered. Thus, the choice of target thresholds and of the locations across which the aggregation is done must be done carefully. 

\subsection{Preservation of dominance}\label{subsec:preservation-dom}

We investigate desirable properties when aggregating ROC or PR curves and the related conditions on parameterization strategies. The first desirable property when aggregating ROC or PR curves is that dominance is preserved. 

\begin{definition}
    \label{def:preservation-dominance}
    Let $A$ and $B$ be two forecasters issuing forecasts across $d$ locations. Let their location-wise ROC (PR) curves be associated with families of contingency tables. Let $\gamma_A$ and $\gamma_B$ be parameterization strategies for forecasters $A$ and $B$, respectively. An aggregation strategy \textit{preserves dominance in terms of ROC (PR) curves} if the aggregated ROC (PR) curve of $A$ dominates the aggregated ROC (PR) curve of $B$, where the aggregated ROC (PR) curves are composed of points of the following form 
    \begin{equation*}
        \left(\FAR^M(\gamma_M(u)), \HR^M(\gamma_M(u))\right)
    \end{equation*}
    \begin{equation*}
        \left(\text{resp. }\left(\Rec^M(\gamma_M(u)),\Pre^M(\gamma_M(u))\right)\right),
    \end{equation*}
    where $u\in[0,1]$ and $M=A,B$.
\end{definition}
When aggregating ROC curves, we obtain the following sufficient condition on the pair of parameterization strategies.

\begin{theorem}\label{thm:preservation-dominance-ROC}
    Let $A$ and $B$ be two forecasters issuing forecasts across $d$ locations. Let their location-wise ROC curves be associated with families of contingency tables. Let $\gamma_A$ and $\gamma_B$ be parameterization strategies for $A$ and $B$, respectively, such that, for all $s=1,\dots,d$ and  $u\in[0,1]$,
    \begin{numcases}{}
        \FAR^A_s(\gamma_{A,s}(u))\leq \FAR^B_s(\gamma_{B,s}(u))\label{eq:locationwise-ineq-FAR}\\
        \HR^A_s(\gamma_{A,s}(u))\geq \HR^B_s(\gamma_{B,s}(u))\label{eq:locationwise-ineq-HR}
    \end{numcases}
     and there exists at least one $u^\ast\in[0,1]$ such that one of the inequality is strict. Then, aggregating interpolated counts preserves dominance in terms of ROC curves.
\end{theorem}
The situation is more delicate for PR curves.
\begin{theorem}\label{thm:preservation-dominance-PR}
    Let $A$ and $B$ be two forecasters issuing forecasts across $d$ locations. Let their location-wise PR curves be associated with families of contingency tables. Let $\gamma_A$ and $\gamma_B$ be parameterization strategies for $A$ and $B$, respectively, satisfying 
    \begin{numcases}{}
        \Rec^A_s(\gamma_{A,s}(u))\geq \Rec^B_s(\gamma_{B,s}(u))\label{eq:locationwise-ineq-Re}\\
        \Pre^A_s(\gamma_{A,s}(u))\geq \Pre^B_s(\gamma_{B,s}(u))\label{eq:locationwise-ineq-Pre}\\
        \FB_s^A(\gamma_{A,s}(u))=\FB_s^B(\gamma_{B,s}(u))\label{eq:locwise-FB-eq}
    \end{numcases}
    and there exists at least one $u^\ast\in[0,1]$ such that one of the inequalities is strict, for all $s=1,\dots,d$. Then, aggregating interpolated counts preserves dominance in terms of PR curves.
\end{theorem}

Eq.~\eqref{eq:locwise-FB-eq} corresponds to a location-wise FB equality. This encapsulates common conditions: central FB and central marginal probabilities. A central FB (regardless of location) is relevant when it is believed that forecasts should be aggregated at decision thresholds corresponding to the same level of rarity relative to the local base rate. A central marginal probability is a stronger condition than that of a central FB, as it is implied by the latter. It corresponds to an interest in events with the same rarity across locations and in aggregating forecasts at thresholds corresponding to unconditional but location-dependent quantiles at the same level. However, it also enforces a central base rate, which is not always desirable. This is the case whenever the definition of the event or its impact is not based on its rarity. For example, since the thresholds for heat warning levels from MeteoSwiss do not depend on location \citep{MeteoSwiss2026}, the base rates of their exceedance differ across Switzerland.

\subsection{Preservation of concavity and achievability}\label{subsec:preservation-concavity-achievability}

The second desirable property of the tuple of parameterization and aggregation strategy we investigate is that concavity is preserved for ROC curves and that achievability is preserved for PR curves. 

\begin{definition}\label{def:preservation-concavity}
    Consider a forecast across $d$ locations and its associated location-wise concave ROC curves. Given a parameterization strategy $\gamma$, an aggregation strategy preserves concavity if the aggregated ROC curve is concave.
\end{definition}

Similarly to the single-location case, the concavity of the aggregated ROC curve is related to the CEP.

\begin{theorem}\label{thm:preservation-concavity}
    Consider a forecaster issuing forecasts across $d$ locations. Let its location-wise concave ROC curves be associated with families of continuous contingency tables. Let $\gamma$ be a parameterization strategy. A necessary and sufficient condition for aggregating interpolated counts to preserve concavity is that the aggregated conditional event probability is non-decreasing in the central parameter $u\in[0,1]$.
\end{theorem}

The aggregated conditional event probability is the CEP of the family of aggregated contingency tables (see Appendix~\ref{appendix:aggregated-CEP}). Similarly, a desirable property when aggregating PR curves is that the aggregation of achievable PR curves yields an achievable PR curve.

\begin{definition}\label{def:preservation-achievability}
    Consider a forecast across $d$ locations and its associated location-wise achievable PR curves. Given a parameterization strategy $\gamma$, an aggregation strategy \textit{preserves achievability} if the aggregated PR curve is an achievable PR curve.
\end{definition}

Since the PR curve associated with a concave ROC curve is an achievable PR curve, the following result is a direct corollary of Theorem~\ref{thm:preservation-concavity}.

\begin{corollary}\label{cor:preservation-achievability}
    Consider a forecaster issuing forecasts across $d$ locations. Let its location-wise achievable PR curves be associated with families of continuous contingency tables. Let $\gamma$ be a parameterization strategy. A necessary and sufficient condition for aggregating interpolated counts to preserve achievability is that the aggregated CEP is non-decreasing in the central parameter $u$.
\end{corollary}

\subsection{Examples}\label{subsec:examples}

\subsubsection{Aggregating counts}\label{subsubsec:aggregating-counts}

The results of Theorems~\ref{thm:preservation-dominance-ROC}, \ref{thm:preservation-dominance-PR}, and \ref{thm:preservation-concavity} and Corollary~\ref{cor:preservation-achievability} remain valid when aggregating counts. However, there are cases in which no parameterization strategies for two forecasts can aggregate counts while simultaneously preserving dominance and concavity. We consider the case of two forecasts, A and B, across two locations. The base rate at location 1 is $p_1=0.9$, and the base rate at location 2 is $p_2=0.1$. Forecasts A and B are characterized by their location-wise ROC curve, represented in Figure~\ref{fig:ce-agg-counts}(a,b).

\begin{figure}
    \centering
    \includegraphics[width=\linewidth]{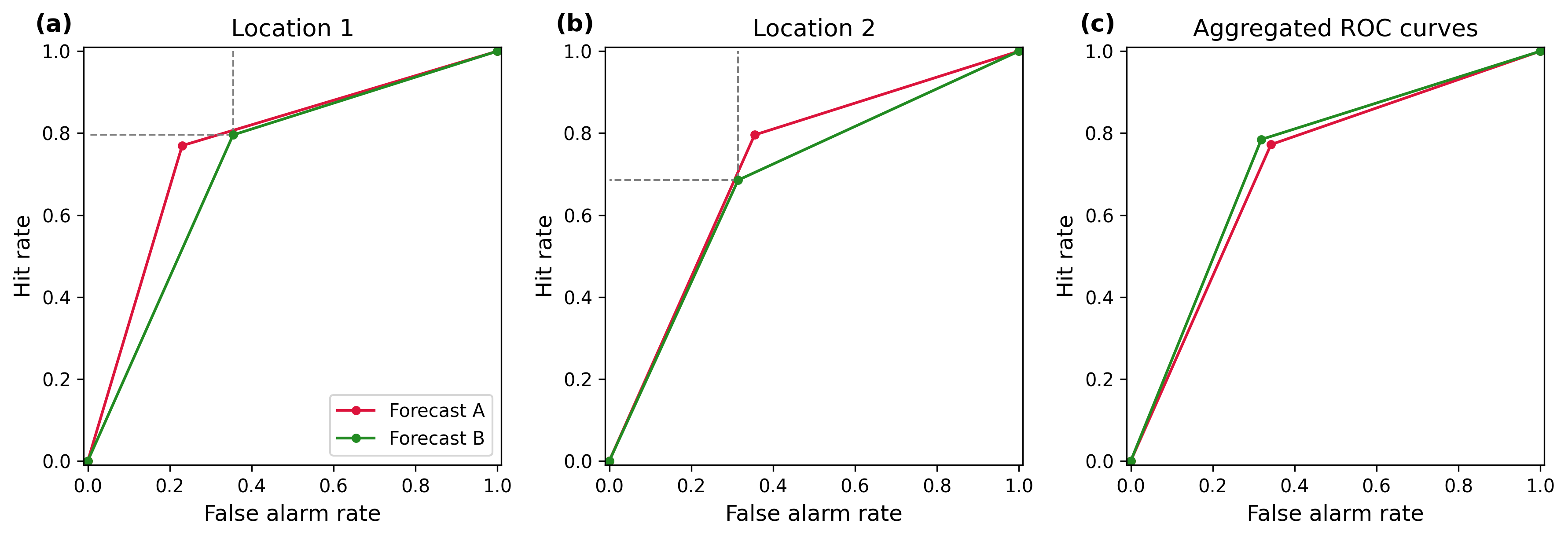}
    \caption{ROC curves of forecasts A and B at location 1 (a), at location 2 (b), and aggregated across both locations (c).}
    \label{fig:ce-agg-counts}
\end{figure}

In this example, the conditions \eqref{eq:locationwise-ineq-FAR} and \eqref{eq:locationwise-ineq-HR} of Theorem~\ref{thm:preservation-dominance-ROC} cannot be satisfied since, at both locations, not all the points of the ROC curve of forecast B (green circles) have a point of the ROC curve of forecast A (red circles) in their upper-left corner (see the grey dotted lines in Figure~\ref{fig:ce-agg-counts}(a,b)). This is not an issue per se since Theorem~\ref{thm:preservation-dominance-ROC} only provides a sufficient condition for the preservation of dominance.\\

This simple setting allows us to compute all the possible combinations of parameterization strategy (not shown). Among all possible parameterization strategies, only one for each forecast preserves concavity. Figure~\ref{fig:ce-agg-counts}(c) shows the comparison of the aggregated ROC curves associated with these strategies. Forecast B dominates that of Forecast A in terms of ROC curves, indicating that, in this case, no pair of parameterization strategies can preserve both concavity and dominance.\\

Note that, in empirical settings, it is not always possible to ensure a location-wise FB equality \eqref{eq:locwise-FB-eq} since there might be multiple forecasts with the same value for one forecaster but not necessarily for the other.

\subsubsection{Aggregating interpolated counts}

We focus on ROC curves, but similar results can be obtained for PR curves due to their equivalence. This section aims to study whether parameterizations exist that preserve both dominance and concavity when paired with the aggregation of interpolated counts. We start by looking at concavity. To preserve concavity, we first need to ensure it at each location. Thus, without loss of generality, we consider ROC curves obtained from ROC diagnostics of location-wise PAV-transformed forecast-observation pairs, and we let these curves be parameterized by $t\in[0,1]$. As the condition provided by Theorem~\ref{thm:preservation-concavity} is general and difficult to grasp, we propose a concrete special case.
\begin{proposition}\label{prop:preservation-concavity-sufficient}
    Consider a location-wise PAV-transformed forecast across $d$ locations and its associated location-wise concave ROC curves. Let $\gamma$ be a parameterization strategy such that for all $u\in[0,1]$, $\gamma_s(u)=f(u)$, where $f:[0,1]\to [0,1]$ is an increasing bijection. Then, the aggregated ROC curve based on aggregating interpolated counts is concave.
\end{proposition}
In particular, the parameterization strategy $\gamma_s(u)=u$ for all $u\in[0,1]$ preserves concavity.

We now investigate the preservation of dominance for aggregating interpolated counts. Given two forecasters $A$ and $B$ across $d$ locations with location-wise concave ROC curves, we study conditions on their parameterization strategies. At each location, we assume that their parameterization is non-decreasing with FAR and HR. For a location $s=1,\dots,d$, given a parameterization strategy $\gamma_A$, the conditions of Theorem~\ref{thm:preservation-dominance-ROC} become
\begin{equation}\label{eq:bounds-param-strat}
    \FAR_s^{B\ (-1)}(\FAR_s^A(\gamma_{A,s}(u))) \leq \gamma_{B,s}(u) \leq \HR_s^{B\ (-1)}(\HR_s^A(\gamma_{A,s}(u)))
\end{equation}
for all $u\in[0,1]$, if the parameterizations of both ROC curves are non-decreasing in FAR and HR. Moreover, there exists $u^\star\in[0,1]$ such that one of the inequalities is strict. Parameterization with a different direction for monotonicity, such as one based on decision thresholds, will lead to an inversion between the lower and upper bounds.

There always exists such a pair of parameterization strategies. Moreover, certain parameterizations of ROC curves satisfy these conditions directly; this is the case when the ROC curves are parameterized by (rescaled) FB or using parallel downward lines in ROC space. Figure~\ref{fig:iso-FB-ROC} illustrates how lines of constant FB span the ROC space.

\begin{figure}
    \centering
    \includegraphics[width=0.5\linewidth]{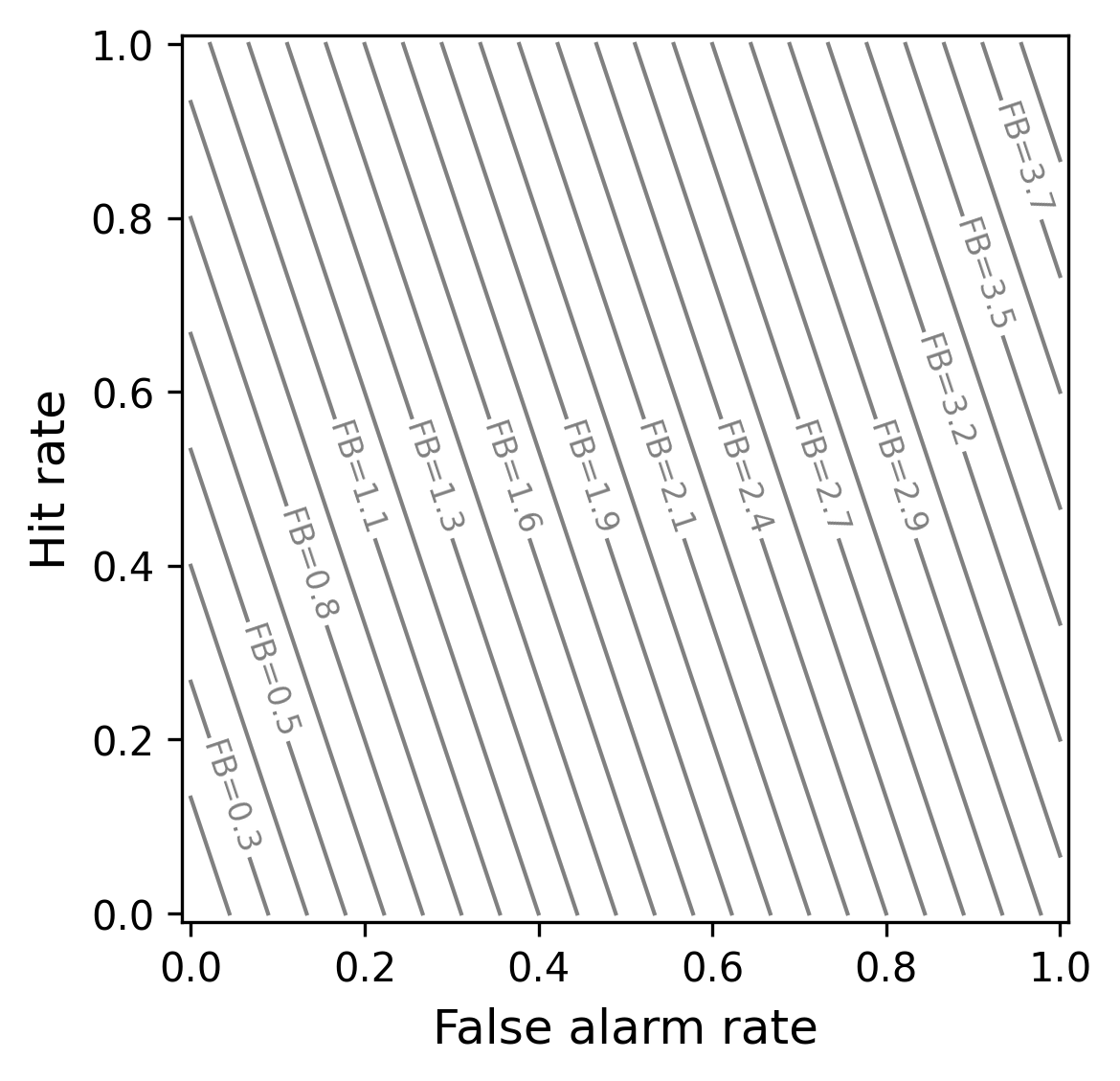}
    \caption{Lines of constant frequency bias in ROC space for a base rate $p=0.2$.}
    \label{fig:iso-FB-ROC}
\end{figure}

In order to simultaneously ensure the preservation of dominance and concavity, we consider ROC curves of two forecasters obtained from ROC diagnostics of location-wise PAV-transformed forecast-observation pairs, and we let these curves be parameterized by a parameter taking values in $[0,1]$. In particular, we assume that these parameterizations are such that $\gamma_{A,s}(u)=u$ and $\gamma_{B,s}(u)=u$ satisfy \eqref{eq:bounds-param-strat} for all $u\in[0,1]$ and $s=1,\dots,d$. Moreover, as a consequence of Proposition~\ref{prop:preservation-concavity-sufficient}, such a pair of parameterization strategies also ensures the preservation of concavity.\\

Thus, the two aforementioned parameterization strategies based on ROC curves being parameterized by FB and based on parallel downward lines with a slope of $-1$ preserve both dominance and concavity. These are referred to as dominance- and concavity-preserving aggregation with FB parametrization (DCP-FB) or with parallel lines parametrization (DCP-PL).

\begin{remark}
    All the above results are stated for the entire ROC and PR curves. However, similar results can be obtained if only subparts of the curves are of interest, provided that the subparts of different curves (at different locations or associated with different forecasts) correspond to decision and target thresholds that satisfy the conditions of the corresponding results.
\end{remark}

Instead of aggregating counts using GraphCast's or the location-scale parameterization strategies, we recommend aggregating interpolated counts and using the family of parameterization strategies presented above, as we illustrate in the following section.

\section{Applications}\label{sec:applications}

In this section, we first illustrate how inappropriate parameterization strategies can lead to misleading conclusions and then compare parameterization strategies on AI-based weather prediction (AIWP) and numerical weather prediction (NWP) forecasts from the WeatherBench 2 dataset \citep{Rasp2024}.

\subsection{Counterexample}\label{subsec:counterexample}

The example in Section~\ref{subsubsec:aggregating-counts} shows that simultaneous preservation of concavity and dominance can be impossible when aggregating counts. The following example illustrates how inappropriate parameterization strategies fail to preserve dominance when aggregating interpolated counts.\\

We use a simple theoretical setting. Let $(X_A,Y)$ and $(X_B,Y)$ be two forecast-observation pairs over two locations. We assume these pairs are jointly Gaussian at each location. The forecast $M\in\{A,B\}$ and the observation at location $s\in\{1,2\}$ are jointly defined as
\begin{equation*}
    (X_{M,s},Y_s)\sim\calN\left(\left[\begin{array}{cc}
        0  \\
        0 
    \end{array}\right],\left[\begin{array}{cc}
        \sigma_{M,s}^2 & \rho_M\sigma_{M,s} \\
        \rho_M\sigma_{M,s} & 1
    \end{array}\right]\right),
\end{equation*}
with $\sigma_{A,1}=\sigma_{A,2}=\sigma_{B,1}=1$ and $\sigma_{B,2}=2$, and $\rho_A=0.85>\rho_B=0.8$. We choose $t_{Y,s}$, the target threshold at location $s\in\{1,2\}$, to be the 95th percentile of the observation distribution. Since $\rho_A$ and $\rho_B$ are both positive, all the location-wise ROC (PR) curves are concave (achievable); see Appendix~\ref{appendix:bivariate-gaussian-cep}. Moreover, since $\rho_A>\rho_B$, $X_A$ dominates $X_B$ in terms of ROC and PR curves at both locations (see Figs.\ref{fig:agg-counterexample-PR}(a,b)~and~\ref{fig:agg-counterexample-ROC}(a,b)).

\begin{figure}[ht]
    \centering
    \includegraphics[width=\linewidth]{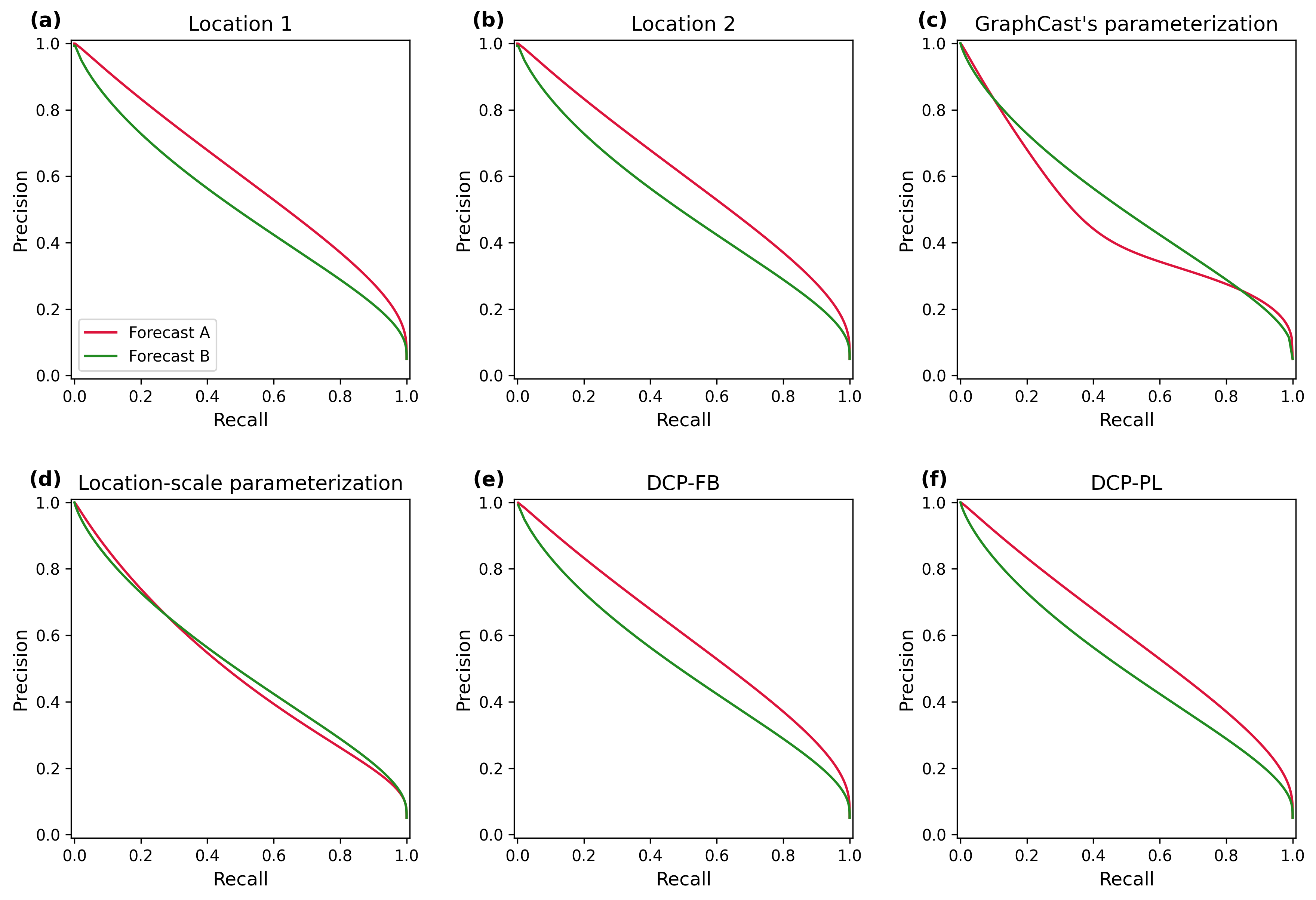}
    \caption{Location-wise PR curves (a,b) and aggregated PR curves using GraphCast's parameterization (c), the location-scale parameterization (d), DCP-FB (e), and DCP-PL (f).}
    \label{fig:agg-counterexample-PR}
\end{figure}

\begin{figure}[ht]
    \centering
    \includegraphics[width=\linewidth]{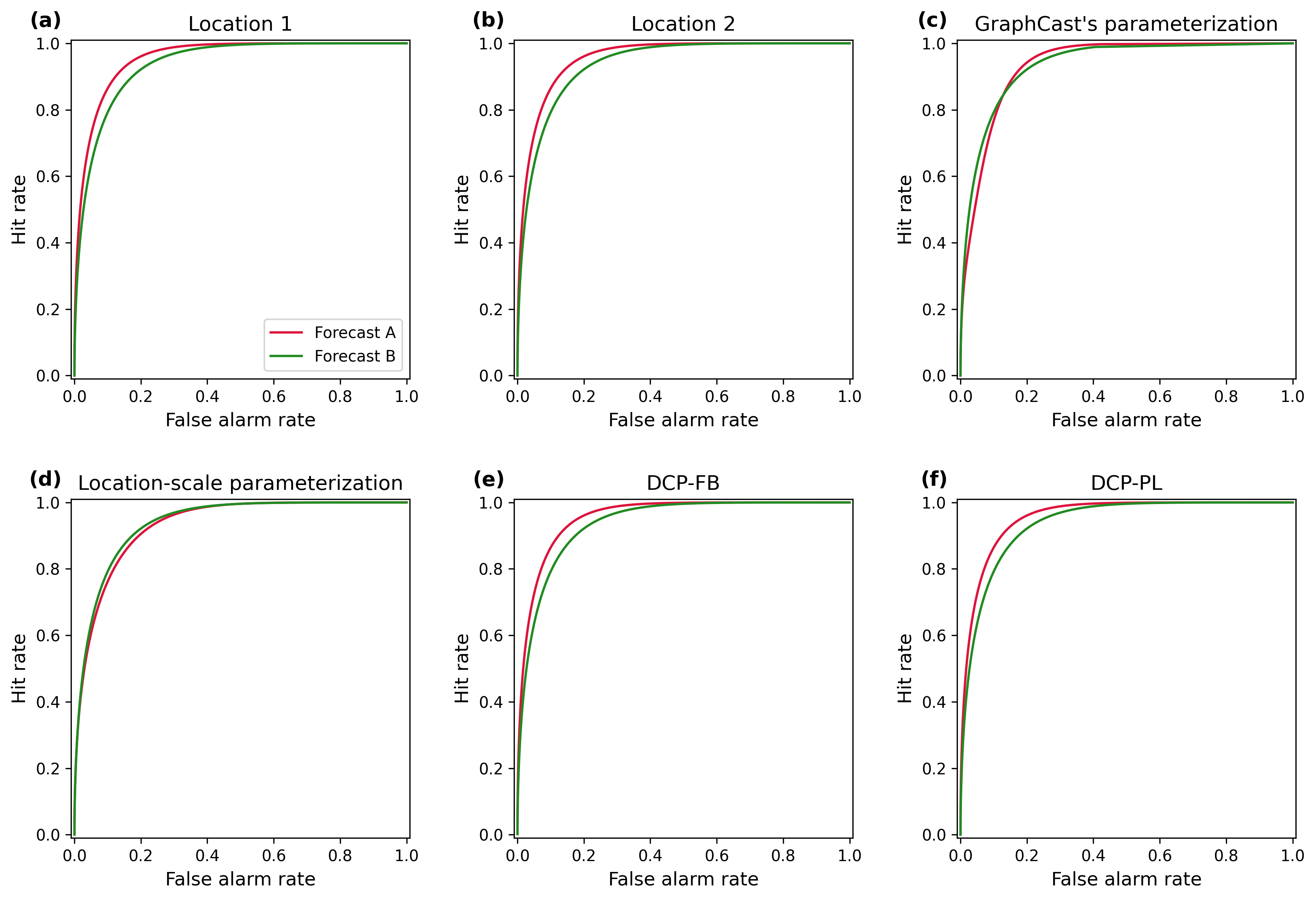}
    \caption{Location-wise ROC curves (a,b) and aggregated ROC curves using GraphCast's parameterization (c), the location-scale parameterization (d), DCP-FB (e), and DCP-PL (f).}
    \label{fig:agg-counterexample-ROC}
\end{figure}

Figures~\ref{fig:agg-counterexample-PR}~and~\ref{fig:agg-counterexample-ROC} show the PR and ROC curves, respectively, at both locations and the aggregated curves obtained using GraphCast's parameterization, the location-scale parameterization, and both DCP-FB and DCP-PL. Note that theoretical quantities have been used (see Appendix~\ref{appendix:bivariate-gaussian-far}), and counts and interpolated counts coincide, as no atom mass is present in the distributions. GraphCast's and the location-scale parameterization strategies fail to preserve dominance, as shown by the crossing of the aggregated curves in panels (c) and (d) of Figs.~\ref{fig:agg-counterexample-PR}~and~\ref{fig:agg-counterexample-ROC}. This is simply due to a difference in scales for $X_B$ at location 2, despite the distributions belonging to a location-scale family.  On the other hand, both DCP-FB and DCP-PL preserve dominance (see Figs.\ref{fig:agg-counterexample-PR}(e,f)~and~\ref{fig:agg-counterexample-ROC}(e,f)). The choice of an appropriate parameterization strategy is crucial, even when aggregating only ROC or PR curves across two locations with the same base rate. Similarly, when the observations' distribution differs at the two locations (i.e., non-stationary), GraphCast's and the location-scale parametrizations still fail to preserve dominance (not shown).

\subsection{WeatherBench 2}

\subsubsection{Dataset description}

We now apply the aforementioned methods to the WeatherBench 2 dataset \citep{Rasp2024}. In particular, we focus on 2m temperature forecasts with lead times ranging from 6 h to 10 days, with a 6 h step initialized at 00:00 and 12:00 UTC daily. To avoid potential comparison issues, we only consider forecasts initialized with the initial conditions of the deterministic configuration (i.e., HRES; \cite{ECMWFdoc}) of the European Centre for Medium-range Weather Forecasting's (ECMWF) Integrated Forecast System (IFS) and use IFS HRES' operational analyses as the ground truth. We consider three models: the ECMWF's IFS HRES; Google DeepMind's GraphCast \citep{Lam2023}; and Huawei's Pangu-Weather \citep{Bi2023}. From now on, IFS HRES, GraphCast, and Pangu-Weather are referred to as HRES, GraphCast, and Pangu, respectively. Forecasts and ground truth are gridded fields with a horizontal resolution of 0.25$^\circ$. All valid times are within the year 2020, avoiding any overlap with the models' training period. Climatological quantities are computed using ERA5 \citep{Hersbach2020} over the period 1979--2019.

\subsubsection{Preservation of dominance}

We now illustrate the preservation of dominance as a key property when aggregating ROC and PR curves using real-world data. We consider two locations, with latitude and longitude (60.25\textdegree, -1.5\textdegree) and (70.25\textdegree, -25.0\textdegree), where HRES dominates GraphCast and Pangu, and GraphCast dominates Pangu in terms of ROC and PR curves (see Figs.~\ref{fig:agg-PR-pair}(a,b) and~\ref{fig:agg-ROC-pair}(a,b)), for a lead time of 2 days and target thresholds corresponding to the 75th climatological percentile for each location, time of day, and month of year. Moreover, the corresponding raw PR and ROC curves satisfy the same dominance relationships.

\begin{figure}[ht]
    \centering
    \includegraphics[width=\linewidth]{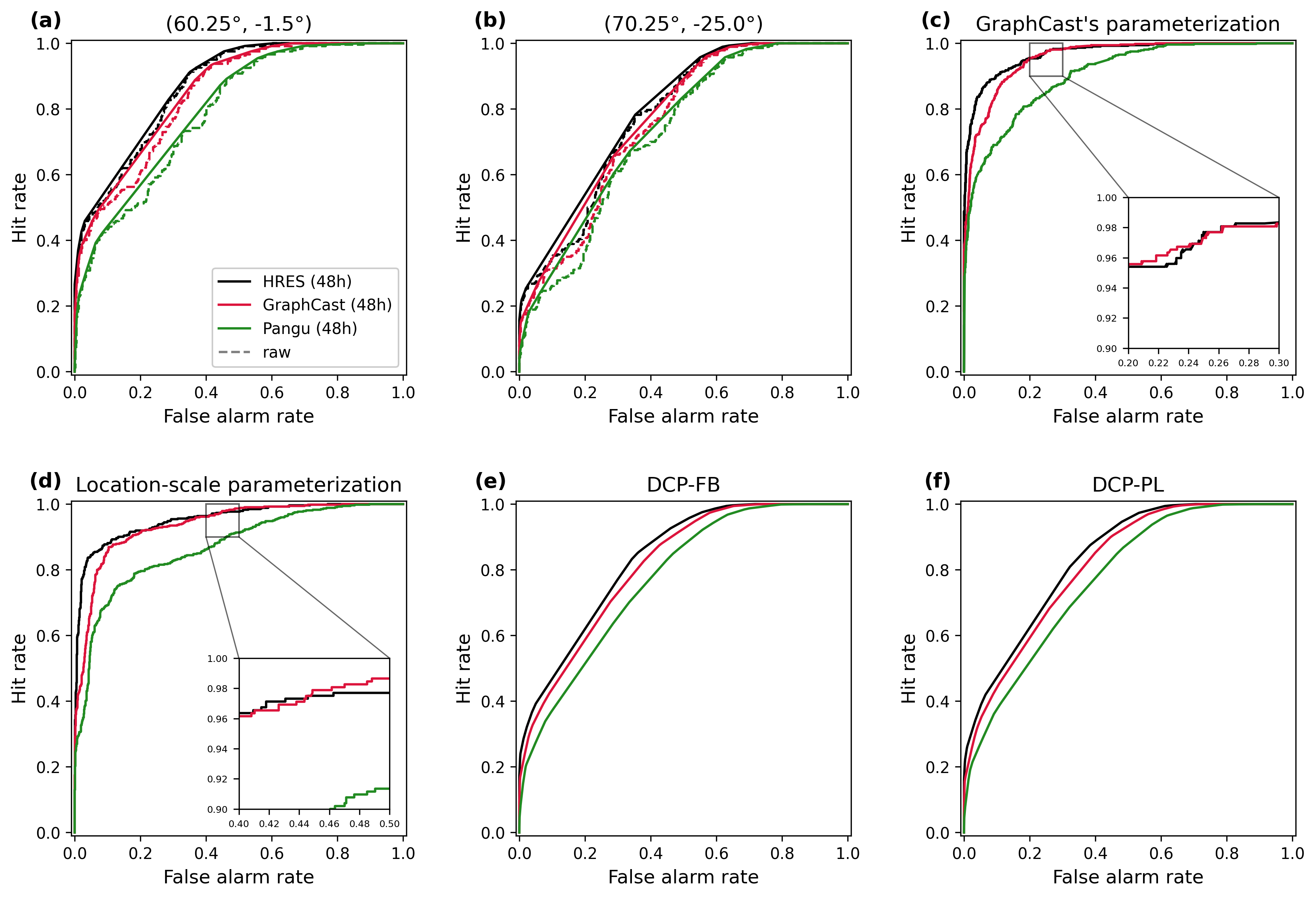}
    \caption{Location-wise raw (dashed lines) and concave (solid) ROC curves (a,b), and aggregated ROC curves using GraphCast's parameterization (c), the location-scale parameterization (d), DCP-FB (e), and DCP-PL (f), for HRES (black), GraphCast (red), and Pangu (green) at a 2-day lead time over (60.25\textdegree, -1.5\textdegree) and (70.25\textdegree, -25.0\textdegree). The target thresholds correspond to the 75th percentiles of ERA5 for each location, time of day, and month of year over 1979-2019.}
    \label{fig:agg-ROC-pair}
\end{figure}

As for the simple counterexample of Section~\ref{subsec:counterexample}, GraphCast's and the location-scale parameterization fail to preserve dominance. In particular, the aggregated PR (ROC) curve based on HRES does not dominate the aggregated PR (ROC) curve based on GraphCast since the curves cross (see Figs.~\ref{fig:agg-PR-pair}(c,d)~and~\ref{fig:agg-ROC-pair}(c,d)). On the other hand, aggregated PR and ROC curves using DCP-FB and DCP-PR preserve the dominance relationships present at the location level (see Figs.~\ref{fig:agg-PR-pair}(e,f)~and~\ref{fig:agg-ROC-pair}(e,f)).

\subsubsection{Extreme temperatures}\label{subsubsec:extreme-temperatures}

We aggregate PR curves across the globe with a focus on extreme temperatures. Once again, we illustrate the impact of the parameterization strategy on the resulting curves. We investigate extreme temperatures using two definitions: exceedances of the 98th climatological percentile, as in \citet{Lam2023}, and record-breaking events relative to the training period, as in \citet{Zhang2026}. However, with both definitions of extreme temperature, there is no location-wise dominance relationship that holds across all locations; thus, eventual dominance relationships among aggregated curves cannot be interpreted as dominance per se.

\paragraph{98th percentile exceedances.}
Following \citet{Lam2023}, the target thresholds correspond to the 98th climatological percentile for each grid point, time of day, and month of the year. Moreover, only grid points over land are considered, and the test period is restricted to summer months (July, August, and September in the Northern hemisphere and December, January, and February in the Southern hemisphere). However, contrary to the original methodology, we rely on ERA5 data from 1979 to 2019 to compute climatological quantities. Figure~\ref{fig:counts-q98} in the Appendix shows the spatial distribution of counts for these target thresholds. The gain ranges used for this worldwide aggregation correspond to those of the original articles (i.e., $g(u)\in[0.8, 4.5]$ and $g(u)\in[-1.5, 1.5]$ for GraphCast's and the location-scale parameterizations, respectively). 

\begin{figure}[ht]
    \centering
    \includegraphics[width=0.7\linewidth]{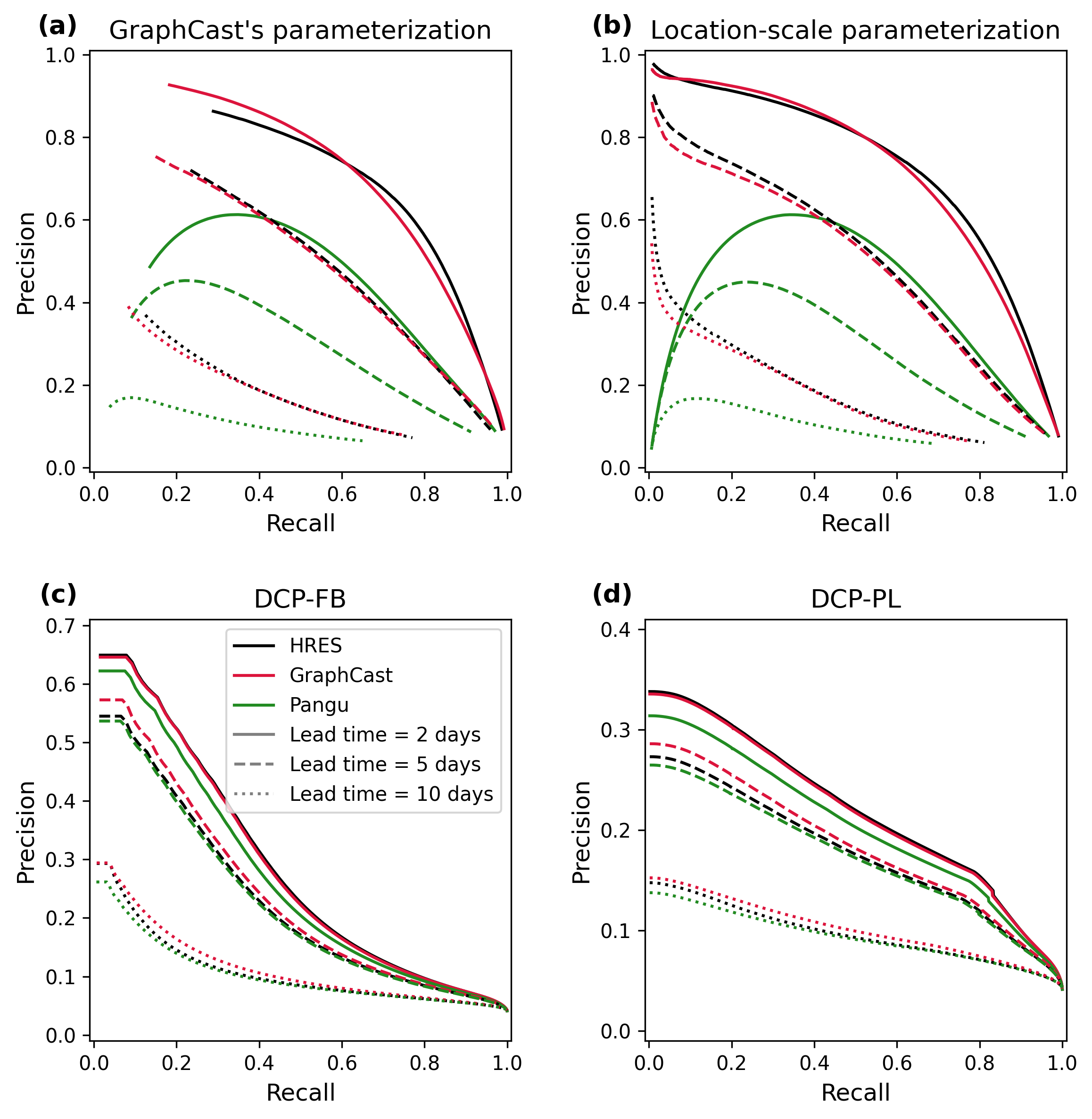}
    \caption{Aggregated PR curves across the globe using GraphCast's parameterization (a), the location-scale parameterization (b), DCP-FB (c), and DCP-PL (d) for HRES (black), GraphCast (red), and Pangu (green) at lead times of 2 days (solid lines), 5 days (dashed lines), and 10 days (dotted lines) over land for summer months, for exceedances of the 98th climatological percentile.}
    \label{fig:agg-PR-q98}
\end{figure}

Figure~\ref{fig:agg-PR-q98} compares forecasts for lead times of 2, 5, and 10 days using GraphCast's parameterization, the location-scale parameterization, DCP-FB, and DCP-PL. First of all, it appears that the different parameterization strategies lead to different ranges of precision values and different decreasing trends of precision as recall increases. GraphCast's and the location-scale parameterization strategies show that Pangu has a notably lower aggregated PR curve than that of HRES and GraphCast for the same lead time. This even leads to the aggregated PR curves of Pangu crossing that of HRES and GraphCast at longer lead times. On the other hand, both DCP-FB and DCP-PL show no crossing between PR curves corresponding to distinct lead times. Nonetheless, Pangu's aggregated PR curves generally appear below those of HRES and GraphCast at the same lead time. DCP-FB and DCP-PL have a smaller range of precision values. This could be due to the fact that, contrary to DCP-FB and DCP-PL, both GraphCast's and the location-scale parameterization strategy can inflate precision values when aggregated, as already visible in Figure~\ref{fig:agg-PR-pair}.

\paragraph{Record-breaking events.}

Following \citet{Zhang2026}, we focus now on target thresholds corresponding to the maximum recorded temperature in the climatology for each grid point and month of the year. Moreover, only grid points over land (excluding Antarctica) are considered. However, contrary to the original methodology, we rely on ERA5 data from 1979 to 2019 to compute climatological quantities, including maximum recorded values. Figure~\ref{fig:counts-max} in the Appendix shows the spatial distribution of counts for these target thresholds. The gain ranges used for this worldwide aggregation correspond to those of the original articles (i.e., $g(u)\in[0.8, 4.5]$ and $g(u)\in[-1.5, 1.5]$  for GraphCast's and the location-scale parameterizations, respectively). 

\begin{figure}[ht]
    \centering
    \includegraphics[width=0.7\linewidth]{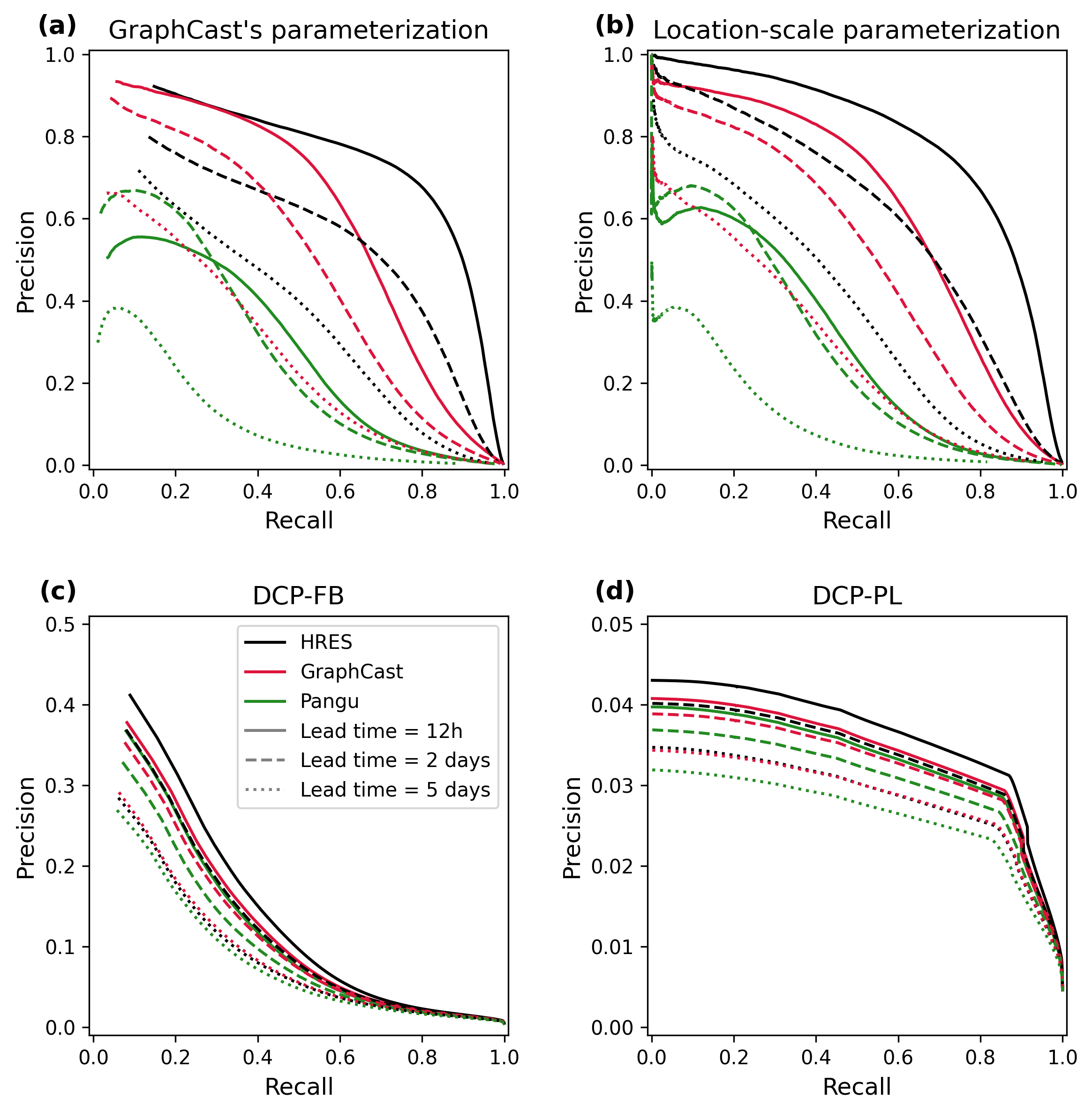}
    \caption{Aggregated PR curves across the globe using GraphCast's parameterization (a), the location-scale parameterization (b), DCP-FB (c), and DCP-PL (d) for HRES (black), GraphCast (red), and Pangu (green) at lead times of 12 h (solid lines), 2 days (dashed lines), and 5 days (dotted lines) over land (excluding Antarctica), for record-breaking events.}
    \label{fig:agg-PR-max}
\end{figure}

Figure~\ref{fig:agg-PR-max} compares forecasts for lead times of 12 h, 2, and 5 days using GraphCast's parameterization, the location-scale parameterization, DCP-FB, and DCP-PL. Here again, different parameterization strategies yield different ranges of precision and distinct decreasing trends in precision as recall increases. For record-breaking target thresholds, GraphCast's and the location-scale parameterization strategies also show that Pangu has a lower aggregated PR curve than that of HRES and GraphCast for the same lead time. This is even stronger here, as Pangu's curve for a 12-h lead time is below that of HRES for all considered lead times and below that of GraphCast for 12 h and 2 days. Both DCP-FB and DCP-PL also show that Pangu's aggregated PR curves are below those of HRES and GraphCast at the same lead time. Moreover, Pangu's curve for a 12-h lead time is below that of HRES for a 2-day lead time. Once again, DCP-FB and DCP-PL have a smaller range of precision values than the two other parameterization strategies

\section{Conclusion}\label{sec:conclusion}

We have provided a framework to aggregate ROC and PR curves. In particular, we proposed and studied two desirable properties when aggregating ROC or PR curves: the preservation of dominance, and the preservation of concavity (for ROC curves) and achievability (for PR curves). We obtained sufficient conditions and investigated two parameterization strategies satisfying them, namely DCP-FB and DCP-PL. Moreover, we have shown how GraphCast's parameterization \citep{Lam2023} and the location-scale parameterization \citep{Zhang2026} fail to preserve dominance using both theoretical and applied counterexamples. We have also investigated how the choice of parameterization strategy affects the aggregation of PR curves across numerous locations following the settings of Fig.~3D in \citet{Lam2023} and Fig.~3D in \citet{Zhang2026}.

Throughout the article, we focused on spatial forecasts as they are the most popular application. However, the points made remain valid for any multivariate forecasts, including spatial forecasts, temporal forecasts, multivariable forecasts, or any combination of these categories.\\

In addition to binary outcomes, the results could be adapted to ROC curves for linearly ordered outcomes \citep{Gneiting2022}. Moreover, the concept of PR curve could also be extended to such outcomes. The effect of the target threshold was not investigated, but future research could address it by examining the properties of aggregated \textit{ROC movies} \citep{Gneiting2022}.

The aggregation strategies considered aim to summarize information across the entire grid, with each grid point having equal importance. An event predicted in the vicinity of an observed event is counted as a false alarm at the location of the forecast event and a miss at the location of the observed event. This phenomenon is called the \textit{double-penalty effect} \citep{Ebert2008}. \citet{Stein2019} have proposed neighborhood-based contingency tables to circumvent the double penalty effect within neighborhoods of a specified size. Local contingency tables are computed over neighborhoods and are then aggregated into a single contingency table. This last step corresponds to aggregating counts. Hence, our proposed framework could be associated with their approach to account for the double-penalty effect. However, as our results impose conditions on the choice of parameterization strategy, further investigation is needed to unveil how to combine these two approaches.

Moreover, in practice, aggregated scores are often weighted by the grid cell area during spatial aggregation; weighted aggregated (interpolated) counts could be considered to relax the assumption that counts have the same meaning across the considered domain.

\subsection*{Recommendations}

In light of our results and the literature, we issue the following recommendations.\\

In the case of a single location, only concave ROC curves should be compared, and only achievable PR curves should be compared. If the raw curves do not meet these criteria, the forecast should be transformed using the PAV algorithm.

In the case of multiple locations (or times or variables), we recommend that the forecast be location-wise PAV-transformed, and then the ROC (resp. PR) curves parameterized by FB or parallel lines can be aggregated based on aggregating counts and with the parameterization strategy $\gamma_s(u)=u$, for all $u\in[0,1]$. This corresponds to DCP-FB and DCP-PL and ensures the preservation of dominance (Definition~\ref{def:preservation-dominance}), the preservation of concavity (Definition~\ref{def:preservation-concavity}) when aggregating ROC curves, and the preservation of achievability (Definition~\ref{def:preservation-achievability}) when aggregating PR curves.

\section*{Acknowledgements}
RP and SE acknowledge the support of the Swiss National Science Foundation (Grant number 186858). RP and JZ gratefully acknowledge funding from the NCCR CLIM+.

\section*{Conflict of interest}

The authors declare that they have no known competing financial interests or personal relationships that could have appeared to influence the work reported in this paper.

\section*{Authors’ contributions}

\textbf{RP}: conceptualization; methodology; software; formal analysis; data curation; visualization; writing – original draft. \textbf{ZZ}: conceptualization; methodology; writing – review and editing. \textbf{JZ}: conceptualization; methodology; funding acquisition; supervision; writing – review and editing. \textbf{SE}: conceptualization; methodology; funding acquisition; supervision; writing – review and editing.

\section*{Data availability statement}

All data used in this study are publicly available through the WeatherBench 2 dataset \citep{Rasp2024}. The code required to reproduce the analysis and figures is openly available in the Zenodo repository at \url{https://doi.org/10.5281/zenodo.22803022}.

\printbibliography

\appendix

\section{Binary forecasts}\label{appendix:binary-forecasts}

We consider a binary forecast $X\in\{0,1\}$ and a binary observation $Y\in\{0,1\}$. We denote $(x_i, y_i), i=1,\dots,n$ as $n$ realizations of the forecast-observation pair $(X,Y)$. The FAR, the HR, the recall, and the precision are then defined as
\begin{align*}
    \FAR &= \frac{\sum_{i=1}^n \ind{x_i=1}\ind{y_i=0}}{\sum_{i=1}^n \ind{y_i=0}};\\
    \HR = \Rec &= \frac{\sum_{i=1}^n \ind{x_i=1}\ind{y_i=1}}{\sum_{i=1}^n \ind{y_i=1}};\\
    \Pre &= \frac{\sum_{i=1}^n \ind{x_i=1}\ind{y_i=1}}{\sum_{i=1}^n \ind{x_i=1}}.
\end{align*}

In this setting, forecasts correspond to points in ROC space. Nonetheless, they can be compared in ROC space and even ranked when one forecast \textit{dominates} another.

\begin{definition}
    Let $X_A$ and $X_B$ be competing binary forecasts for the same forecasting task. $X_A$ \textit{dominates} $X_B$ \textit{in ROC space} if $\FAR^A\leq \FAR^B$ and $\HR^A\geq \HR^B$, with at least one strict inequality, where $\FAR^M$ and $\HR^M$ are the false alarm rate and the hit rate, respectively, of forecast $X_M$ with $M=A,B$.
\end{definition}

Similarly, binary forecasts correspond to points in PR space. \textit{Dominance in PR space} can be defined to compare and rank binary forecasts.
\begin{definition}
    Let $X_A$ and $X_B$ be competing binary forecasts. $X_A$ \textit{dominates} $X_B$ \textit{in PR space} if 
    $\Rec^A\geq \Rec^B$ and $\Pre^A\geq \Pre^B$, with at least one strict inequality, where $\Rec^M$ and $\Pre^M$ are the recall and the precision, respectively, of forecast $X_M$ with $M=A,B$.
\end{definition}

Binary forecasts can be issued across multiple locations. We consider forecasts and observations over $d$ spatial locations. Within this setting, let $X\in\{0,1\}^d$ be a forecast and $Y\in\{0,1\}^d$ an observation. We denote $(x_i,y_i)$, $i=1,\dots,n$ as $n$ realizations of the forecast-observation pair $(X,Y)$. A desirable property when spatially aggregating binary forecasts in ROC or PR space is the preservation of dominance.
\begin{definition}
    Let $X_A$ and $X_B$ be two binary forecasts such that $X_A$ dominates $X_B$ in ROC (PR) space across $d$ locations. An aggregation strategy \textit{preserves dominance in ROC (PR) space} if $X_A$ dominates $X_B$ in ROC (PR) space after aggregation across the $d$ locations.
\end{definition}

Since both the forecast and the observation are binary, there is no freedom induced by the choice of decision and target thresholds. Hence, we obtain the following general result.

\begin{theorem}\label{thm:ROC_PR-space-dom-opt1}
    The aggregation strategy based on aggregating counts :
    \begin{itemize}
        \item preserves dominance in ROC space;
        \item fails to preserve dominance in PR space.
    \end{itemize}
\end{theorem}

\begin{proof}[Proof of Theorem~\ref{thm:ROC_PR-space-dom-opt1}.]
    
Let $X_A$ and $X_B$ be two binary forecasts over $d$ locations. 

\begin{proof}[Proof in ROC space.]
    At each location $s=1,\dots,d$, $X_A$ dominates $X_B$ in ROC space, i.e.,
\begin{equation*}
    \begin{cases}
        \FAR_s^A \leq \FAR_s^B\\
        \HR_s^A\geq \HR^B_s
    \end{cases},
\end{equation*}
with at least one strict inequality.\\
    
Since the denominator in $\FAR^M_s$ and $\HR^M_s$ does not depend on the forecast $M=A,B$, the inequalities hold for the numerators. It corresponds to the following conditions 
\begin{equation*}
    \begin{cases}
        \sum_{i=1}^n \ind{x_{i,s}^A=1}\ind{y_{i,s}=0} \leq \sum_{i=1}^n \ind{x_{i,s}^B=1}\ind{y_{i,s}=0}\\
        \sum_{i=1}^n \ind{x_{i,s}^A=1}\ind{y_{i,s}=1} \geq \sum_{i=1}^n \ind{x_{i,s}^B=1}\ind{y_{i,s}=1}
    \end{cases},
\end{equation*}
with at least one strict inequality, for all $s=1,\dots,d$. The aggregated FAR and HR are given by \eqref{eq:FAR-interp-counts} and \eqref{eq:HR_Re-interp-counts}, where it can be noted that the denominators are independent of the forecast $M=A,B$. Thus,
\begin{equation*}
    \begin{cases}
        \FAR^A \leq \FAR^B\\
        \HR^A \geq \HR^B
    \end{cases},
\end{equation*}
with at least one strict inequality, where $\FAR^M$ and $\HR^M$ are the aggregated false alarm rate and the aggregated hit rate, respectively, of forecast $M=A,B$.
\end{proof}

\begin{proof}[Proof in PR space.]
    Since recall and HR have the same definition, the inequalities for the aggregated recall hold. However, the aggregated precision can cause issues. Example~\ref{ce:dom-PR-space} shows that \eqref{eq:HR_Re-interp-counts} and \eqref{eq:Pre-interp-counts} do not always preserve dominance in PR space.
\end{proof}
\end{proof}

\subsection{Sufficient conditions for dominance preservation in PR space}

Under certain conditions, aggregating counts can preserve dominance in PR space.

\begin{proposition}\label{prop:PRspace-dom-fb}
    A sufficient condition for aggregating counts to preserve dominance in PR space is that, at any location, the binary forecasts $X_A$ and $X_B$ have the same frequency bias. Formally, $\FB_s^A=\FB_s^B$, where $\FB_s^M$ is the frequency bias of forecast $X_M$, with $M=A,B$, at location $s$, for all $s=1,\dots,d$.
\end{proposition}

\begin{proof}[Proof for Proposition~\ref{prop:PRspace-dom-fb}.]
    Assume that at any location $s=1,\dots,d$, the following equality on FB holds
    \begin{align*}
        & \qquad \FB_s^A=\FB^B_s\\
        &\Leftrightarrow \frac{\sum_{i=1}^n \ind{x_{i,s}^A=1}}{\sum_{i=1}^n \ind{y_{i,s}=1}} = \frac{\sum_{i=1}^n \ind{x_{i,s}^B=1}}{\sum_{i=1}^n \ind{y_{i,s}=1}}\\
        &\Leftrightarrow \sum_{i=1}^n \ind{x_{i,s}^A=1} = \sum_{i=1}^n \ind{x_{i,s}^B=1}.
    \end{align*}
    Hence, the fact that $X_A$ dominates $X_B$ in PR space across locations leads to
    \begin{equation*}
    \begin{cases}
        \sum_{i=1}^n \ind{x_{i,s}^A=1}\ind{y_{i,s}=1} \geq \sum_{i=1}^n \ind{x_{i,s}^B=1}\ind{y_{i,s}=1}\\
        \Rec^A_s \geq \Rec^B_s
    \end{cases},
\end{equation*}
with at least one strict inequality. Since the denominator in \eqref{eq:Pre-interp-counts} is independent of the forecast $M=A,B$, then we obtain
\begin{equation*}
    \begin{cases}
        \Pre^A \geq \Pre^B\\
        \Rec^A \geq \Rec^B
    \end{cases},
\end{equation*}
with at least one strict inequality.
\end{proof}

Note that there is no equivalence to the preservation of concavity for binary forecasts, as concavity cannot be defined.

\section{Pool-Adjacent-Violators algorithm}\label{appendix:PAV}

Let $G_{i:j}$ denote the group from index $i$ to index $j$ (both included). Let $t_Y$ be a fixed target threshold.

\begin{algorithm}
    \caption{\citep{Dimitriadis2024} PAV algorithm}\label{alg:PAV}
    \begin{algorithmic}
        \Require $(x_1,y_1),\ldots,(x_n,y_n)\in\bbR^2$ where $x_1\leq\ldots\leq x_n$
        \Ensure calibrated forecast values $\hat{x}_1,\ldots,\hat{x}_n$\\
        partition into groups $G_{1:1},\ldots,G_{n:n}$ and let $\hat{x}_i=\onefun_{y_i>t_Y}$ for $i=1,\ldots,n$
        
        \While{there are groups $G_{k:i}$ and $G_{(i+1):l}$ such that $\hat{x}_1\leq\ldots\leq \hat{x}_i$ and $\hat{x}_i>\hat{x}_{i+1}$}
        \State merge $G_{k:i}$ and $G_{(i+1):l}$ and let $\hat{x_i}=\frac{1}{l-k+1}\sum_{j=k}^l y_j$ for $i=k,\ldots,l$
        \EndWhile
    \end{algorithmic}
\end{algorithm}

\section{Dominance equivalence}\label{appendix:dominance-equiv}

\begin{lemma}\label{lemma:dominance-equiv-ROC}
    Let $\hat{R}_A, \hat{R}_B \in \mathcal{R}$ be concave ROC curves. The following two definitions of dominance between ROC curves are equivalent
    \begin{enumerate}[1)]
        \item for all $x_B\in[0,1]$, there exists $x_A\in[0,1]$ such that
        \[
             \begin{cases}
                x_A \leq x_B\\
                \hat{R}_A(x_A) \geq \hat{R}_B(x_B)
            \end{cases},
        \]
        and there exists at least one $x_A^\ast, x_B^\ast\in[0,1]$ satisfying $x_A^*\leq x_B^*$ and $\hat{R}_A(x_A^\ast)\geq \hat{R}_B(x_B^*)$ with at least one strict inequality.
        \item for all $x\in[0,1]$,  $\hat{R}_A(x)\geq \hat{R}_B(x)$. Moreover, there is at least one $x^\ast\in[0,1]$ such that the inequality is strict.
    \end{enumerate}
\end{lemma}

\begin{proof}
    Let $\hat{R}_A, \hat{R}_B \in \mathcal{R}$ be concave ROC curves.

    \noindent\textit{Proof of $1)\Rightarrow 2)$}.
    We prove the contrapositive (i.e., $\text{not } 2) \Rightarrow\text{not }1)$).
    
    \paragraph{Case 1:} There is no strict inequality. This is equivalent to $\hat{R}_A=\hat{R}_B$.
    
    \paragraph{Case 2:} There exist $x^\ast\in[0,1]$ such that $\hat{R}_A(x^\ast)<\hat{R}_B(x^\ast)$. Let $x_A\in[0,1]$ such that $x_A\leq x^\ast$. Since $\hat{R}_A$ is non-decreasing, we have $\hat{R}_A(x_A)\leq \hat{R}_A(x^\ast)<\hat{R}_B(x^\ast)$. Thus, all the points of $\hat{R}_A$ on the left of $(x^*,\hat{R}_B(x^*))$ are strictly below it.

    \noindent\textit{Proof of $2)\Rightarrow 1)$}. It is trivial by definition.
\end{proof}

Similarly, we can define the dominance of PR curves in two equivalent ways. The proof is omitted since the same proving methodology as in Lemma~\ref{lemma:dominance-equiv-ROC} can be applied.

\begin{lemma}\label{lemma:dominance-equiv-PR}
    Let $\hat{P}_A, \hat{P}_B \in \mathcal{P}$ be the PR curves associated with the forecasts $A$ and $B$, respectively. The two following definitions of dominance are equivalent
    \begin{enumerate}[1)]
        \item for all $x_B\in[0,1]$, there exists $x_A\in[0,1]$ such that
        \[
             \begin{cases}
                x_A \geq x_B\\
                \hat{P}_A(x_A) \geq \hat{P}_B(x_B)
            \end{cases},
        \]
        and for all $x_B\in[0,1]$ there exists at least one point $x_A^*\in[0,1]$ satisfying $x_A^*\geq x_B$ and $\hat{P}_A(x_A^\ast)\geq \hat{P}_B(x_B)$ with at least one strict inequality.
        \item for all $x\in[0,1]$, $\hat{P}_A(x)\geq \hat{P}_B(x)$. Moreover, there is at least one $x^\ast\in[0,1]$ such that the inequality is strict.
    \end{enumerate}
\end{lemma}

\section{Proofs}

\subsection{Preservation of dominance}

\subsubsection{Theorem~\ref{thm:preservation-dominance-ROC}}

\begin{proof}[Proof of Theorem~\ref{thm:preservation-dominance-ROC}.]
    
Let $\left\{
    \begin{bmatrix}
        a_s^M(t) & b_s^M(t) \\
        c_s^M(t) & d_s^M(t)
    \end{bmatrix}, t\in\bar{\bbR}
\right\}
$
with $s=1,\dots,d$ be families of contingency tables associated with the forecast-observation pair $(X_M,Y)$, for $M=A,B$. Then, Eqs. \eqref{eq:locationwise-ineq-FAR} and \eqref{eq:locationwise-ineq-HR} become
\begin{equation*}
    \begin{cases}
        \cfrac{b_s^A(\gamma_{A,s}(u))}{b_s^A(\gamma_{A,s}(u))+d_s^A(\gamma_{A,s}(u))} \leq \cfrac{b_s^B(\gamma_{B,s}(u))}{b_s^B(\gamma_{B,s}(u))+d_s^B(\gamma_{B,s}(u))}\\
        \cfrac{a_s^A(\gamma_{A,s}(u))}{a_s^A(\gamma_{A,s}(u))+c_s^A(\gamma_{A,s}(u))} \geq \cfrac{a_s^B(\gamma_{B,s}(u))}{a_s^B(\gamma_{B,s}(u))+c_s^B(\gamma_{B,s}(u))}
    \end{cases}
\end{equation*}
and since $b_s^M(t)+d_s^M(t)$ and $a_s^M(t)+c_s^M(t)$ are both independent of $t$ and $M$, it leads to
\begin{equation*}
    \begin{cases}
        b_s^A(\gamma_{A,s}(u)) \leq b_s^B(\gamma_{B,s}(u))\\
        a_s^A(\gamma_{A,s}(u)) \geq a_s^B(\gamma_{B,s}(u))
    \end{cases}.
\end{equation*}

Summing over all locations and dividing by a positive constant, we obtain 
\begin{equation*}
    \begin{cases}
        \cfrac{\sum_{s=1}^d b_s^A(\gamma_{A,s}(u))}{\sum_{s=1}^d \left( b_s^A(\gamma_{A,s}(u))+d_s^A(\gamma_{A,s}(u))\right)} \leq \cfrac{\sum_{s=1}^d b_s^B(\gamma_{B,s}(u))}{\sum_{s=1}^d \left( b_s^B(\gamma_{B,s}(u))+d_s^B(\gamma_{B,s}(u))\right)}, \\
        \cfrac{\sum_{s=1}^d a_s^A(\gamma_{A,s}(u))}{\sum_{s=1}^d \left( a_s^A(\gamma_{A,s}(u))+c_s^A(\gamma_{A,s}(u))\right)} \geq \cfrac{\sum_{s=1}^d a_s^B(\gamma_{B,s}(u))}{\sum_{s=1}^d \left(a_s^B(\gamma_{B,s}(u))+c_s^B(\gamma_{B,s}(u))\right)},
    \end{cases} 
\end{equation*}
which is 
\begin{equation*}
    \begin{cases}
        \FAR^A(\gamma_A(u))\leq \FAR^B(\gamma_B(u)),\\
        \HR^A(\gamma_A(u))\geq \HR^B(\gamma_B(u)).
    \end{cases}
\end{equation*}

This reasoning remains valid when considering $u^\star\in[0,1]$ such that one of the inequalities is strict.
\end{proof}

\subsubsection{Theorem~\ref{thm:preservation-dominance-PR}}

To show that preserving dominance between PR curves differs from preserving it between ROC curves, we provide an example showing how aggregating precision across locations can fail to preserve dominance.

\begin{example}\label{ce:dom-PR-space}
    Consider two binary forecasts, $A$ and $B$, issued over two locations. At location $s=1$, let
\begin{align*}
    \Pre_1^A &=\frac{a_1^A}{a_1^A+b_1^A} = \frac{4}{5};\\
    \Pre_1^B &=\frac{a_1^B}{a_1^B+b_1^B} = \frac{6}{8}.
\end{align*}
At location $s=2$, let
\begin{align*}
    \Pre_2^A &= \frac{a_2^A}{a_2^A+b_2^A} = \frac{2}{8};\\
    \Pre_2^B &= \frac{a_2^B}{a_2^B+b_2^B} = \frac{1}{5}.
\end{align*}
We have that $A$ dominates $B$ in PR space at both locations since
\begin{equation*}
    \begin{cases}
        \Pre_1^A = \frac{4}{5} &> \Pre_1^B = \frac{3}{4}\\
        \Pre_2^A = \frac{1}{4} &> \Pre_2^B = \frac{1}{5}
    \end{cases}.
\end{equation*}

However,
\begin{align*}
    \Pre^A &= \frac{a_1^A+a_2^A}{a_1^A+b_1^A+a_1^A+b_1^A} = \frac{4+2}{5+8} = \frac{6}{13};\\
    \Pre^B &= \frac{a_1^B+a_2^B}{a_1^B+b_1^B+a_1^B+b_1^B} = \frac{6+1}{8+5} = \frac{7}{13}.
\end{align*}
Hence, $\Pre^A<\Pre^B$.
\end{example}

Since recall and HR have the same definition, the inequalities for the aggregated recall hold. However, the aggregated precision can cause issues. Example~\ref{ce:dom-PR-space} shows that aggregating interpolated counts does not always preserve dominance in PR space.\\

\begin{proof}[Proof of Theorem~\ref{thm:preservation-dominance-PR}.]
    Note that the inequalities for the aggregated recall hold. Let $\left\{
        \begin{bmatrix}
            a_s^M(t) & b_s^M(t) \\
            c_s^M(t) & d_s^M(t)
        \end{bmatrix}, t\in\bar{\bbR}
    \right\}
    $
    with $s=1,\dots,d$ be families of contingency tables associated with the forecast-observation pair $(X_M,Y)$, for $M=A,B$. Eq.~\eqref{eq:locwise-FB-eq} becomes
    \begin{align*}
        & \qquad \cfrac{a_s^A(\gamma_{A,s}(u))+b_s^A(\gamma_{A,s}(u))}{a_s^A(\gamma_{A,s}(u))+c_s^A(\gamma_{A,s}(u))} = \cfrac{a_s^B(\gamma_{B,s}(u))+b_s^B(\gamma_{B,s}(u))}{a_s^B(\gamma_{B,s}(u))+c_s^B(\gamma_{B,s}(u))}\\
        &\Leftrightarrow \  a_s^A(\gamma_{A,s}(u))+b_s^A(\gamma_{A,s}(u)) = a_s^B(\gamma_{B,s}(u))+b_s^B(\gamma_{B,s}(u)),
    \end{align*}
    since $b_s^M(t)+d_s^M(t)$ and $a_s^M(t)+c_s^M(t)$ are both independent of $t$ and $M$. Thus, Eq.~\eqref{eq:locationwise-ineq-Pre} leads to $a_s^A(\gamma_{A,s}(u))\geq a_s^B(\gamma_{B,s}(u))$, for all $s=1,\dots,d$ and $u\in[0,1]$. Finally, summing the inequalities and dividing by the same nonzero quantities on both sides yields the desired inequality for aggregated precisions.
    
    This reasoning remains valid when considering $u^\star\in[0,1]$ such that one of the inequalities is strict.
\end{proof}

\subsection{Preservation of concavity and achievability}

\subsubsection{CEP of a family of continuous contingency tables}\label{appendix:aggregated-CEP}

This serves as a preamble to the following proof and details how to compute the CEP of a family of continuous contingency tables. We consider a family of continuous contingency tables,
\begin{equation*}
\left\{
    \begin{bmatrix}
        a(t) & b(t) \\
        c(t) & d(t)
    \end{bmatrix}, t\in\bar{\bbR}\right\},
\end{equation*}
indexed by $t\in\bar{\bbR}$. The ROC curve associated with this family corresponds to points of the form 
\begin{equation*}
    \left\{\left(\frac{b(t)}{b(t)+d(t)},\frac{a(t)}{a(t)+c(t)}\right), t\in\bar{\bbR} \right\},
\end{equation*}
and its raw ROC diagnostic can be obtained by removing horizontal and vertical segments of the curve (e.g., by restricting the values taken by $t$). Additionally, $(a(t)+c(t))/(a(t)+b(t)+c(t)+d(t))$ is constant and $t\mapsto (c(t)+d(t))/(a(t)+b(t)+c(t)+d(t))$, $t\in\bbR$, is a cumulative distribution function on $\bbR$. Thus, they can characterize the marginal distribution of a binary random variable and a real-valued random variable, respectively. 
Given Corollary 1 in \citet{GneitingVogel2022}, a pair of a continuous forecast and a binary observation can be characterized by a family of continuous contingency tables. Moreover, a CEP can thus be associated with it.

\subsubsection{Theorem~\ref{thm:preservation-concavity}}

\begin{proof}[Proof of Theorem~\ref{thm:preservation-concavity}.]
    We consider the concave ROC curves associated with location-wise PAV-transformed forecasts, rather than another parameterization of the concave ROC curve, without loss of generality.

     Let $(\hat{X}_{1},Y_1),\ldots,(\hat{X}_{d},Y_d)\in[0,1]\times\{0,1\}$ be location-wise PAV-transformed forecasts. Let $t_{Y,s}$ be the target threshold at location $s=1,\dots,d$. At every location $s=1,\dots,d$, the ROC curve is concave and the forecasts are calibrated. Thus, the local CEP is 
    \begin{equation}\label{eq:CEP-location}
        \mathrm{CEP}_s(x) = \frac{\widehat{\Pr}_s(\hat{X}_s=x,Y_s>t_{Y,s})}{\widehat{\Pr}_s(\hat{X}_s=x)} = x,
    \end{equation}
    for all $x\in[0,1]$ such that the denominator is non-zero, where $\widehat{\Pr}_s$ is the empirical distribution of $(\hat{x}_{s,1},y_{s,1}),\dots,(\hat{x}_{s,n},y_{s,n})$ and $(\hat{X}_s,Y_s)$ is a random draw from $\widehat{\Pr}_s$. As a consequence, we have that 
    \begin{equation}\label{eq:calibrated-CEP}
         \widehat{\Pr}_s(\hat{X}_s=x,Y_s>t_{Y,s}) = x\cdot \widehat{\Pr}_s(\hat{X}_s=x),
    \end{equation}
    for all $x\in[0,1]$ such that the denominator in \eqref{eq:CEP-location} is non-zero.

    The equivalence between the concavity of a ROC curve and the CEP being non-decreasing holds only for a single location at a time \citep{GneitingVogel2022, Dumbgen2023}. The aggregated ROC curve is the ROC curve associated with the family of aggregated contingency tables. Hence, it is also the ROC curve associated with a pair of random variables (i.e., a virtual forecast-observation pair) characterized by the aforementioned family. The aggregated CEP (i.e., the CEP of the virtual forecast-observation pair) is given by
    \begin{equation*}
        \mathrm{CEP}(u) := \frac{\sum_{s=1}^d \widehat{\Pr}_s(\hat{X}_s=\gamma_s(u),Y_s>t_{Y,s})}{\sum_{s=1}^d \widehat{\Pr}_s(\hat{X}_s=\gamma_s(u))}.
    \end{equation*}
    Hence, the concavity of the aggregated ROC curve is characterized by the aggregated CEP being non-decreasing. Substituting the numerator using \eqref{eq:calibrated-CEP} leads to 
    \begin{equation}\label{eq:aggregated-CEP}
        \mathrm{CEP}(u) = \frac{\sum_{s=1}^d \widehat{\Pr}_s(\hat{X}_s=\gamma_s(u)) \cdot \gamma_s(u)}{\sum_{s=1}^d \widehat{\Pr}_s(\hat{X}_s=\gamma_s(u))}
    \end{equation}
    and the aggregated ROC curve is concave if and only if $\mathrm{CEP(u)}$ is non-decreasing with $u$, for $u\in[0,1]$ such that the denominator in \eqref{eq:aggregated-CEP} is non-zero.
\end{proof}

\subsubsection{Corollary~\ref{cor:preservation-achievability}}

Corollary~\ref{cor:preservation-achievability} is a direct application of Theorem~\ref{thm:preservation-concavity} and the fact that a family of contingency tables with an achievable PR curve also has a concave ROC curve (and vice versa).

\subsection{Examples}

\begin{proof}[Proof of Proposition~\ref{prop:preservation-concavity-sufficient}.]
    Plugging $\gamma_s(u)=f(u)$ in \eqref{eq:aggregated-CEP} leads to $\mathrm{CEP}(u)=f(u)$ for all $u\in[0,1]$ such that the denominator in the aggregated CEP is non-zero.
\end{proof}

\begin{proof}[Proof of Eq.~\ref{eq:bounds-param-strat}.]
    At location $s$, we consider both a general parameterization $(\FAR_s^M(t), \HR_s^M(t))$, with $t\in\bar{\bbR}$ and the parameterization by FAR where the ROC curve is denoted by $R_s^M\in\mathcal{R}$ for forecaster $M=A,B$.

    Set $t_f(u)=\FAR_s^{B\ (-1)}(\FAR_s^A(\gamma_{A,s}(u)))$, it is the parameter at which the FAR of $B$ matches that of $A$ associated with the parameter $\gamma_{A,s}(u)$ :
    \begin{equation*}
        \FAR_s^{B}(t_f(u))=\FAR_s^A(\gamma_{A,s}(u)).
    \end{equation*}
    The FAR parameterization yields
    \begin{align*}
        \HR_s^{B}(t_f(u)) &= R_s^B(\FAR_s^{B}(t_f(u)));\\
                          &= R_s^B(\FAR_s^{A}(\gamma_{A,s}(u))).
    \end{align*}
    By definition, the fact that $A$ dominates $B$ in terms of ROC curves at location $s$ implies
    \begin{align}
        &\qquad R_s^B(\FAR_s^{A}(\gamma_{A,s}(u))) \leq R_s^A(\FAR_s^{A}(\gamma_{A,s}(u)));\nonumber\\
        &\Leftrightarrow \HR_s^{B}(t_f(u)) \leq \HR_s^{A}(\gamma_{A,s}(u))\label{eq:HR-ineq}.
    \end{align}
    Similarly, set $t_h(u)=\HR_s^{B\ (-1)}(\HR_s^A(\gamma_{A,s}(u)))$, the parameter such that the FAR of $B$ matches that of $A$ at the parameter value $\gamma_{A,s}(u)$: i.e., $\HR_s^{B}(t_h(u))=\HR_s^A(\gamma_{A,s}(u))$. Thus, from \eqref{eq:HR-ineq}, we have $\HR_s^{B}(t_f(u)) \leq \HR_s^{B}(t_h(u))$.     Since the ROC curves of $B$ at locations are concave, $\HR_s^{B\ 
    (-1)}$ is increasing and, thus,
    \begin{equation*}
        t_f(u)\leq t_h(u)
    \end{equation*}
    for all $u\in[0,1]$.

    Moreover, since $t_f(u)$ and $t_h(u)$ are the values of the parameter such that $\FAR_s^B$ matches $\FAR_s^A(\gamma_{A,s}(u))$ and $\HR_s^B$ matches $\HR_s^A(\gamma_{A,s}(u))$, respectively. Assuming that the parameterization of the ROC curves of $B$ is such that increasing parameter values lead to an increasing FAR and HR, the conditions of \eqref{eq:locationwise-ineq-FAR} and \eqref{eq:locationwise-ineq-HR} lead to 
    \begin{equation*}
        t_f(u)\leq \gamma_{B,s}(u) \leq t_h(u)
    \end{equation*}
    for all $u\in[0,1]$.
\end{proof}

\section{Bivariate Gaussian case}

Let a forecast-observation pair $(X,Y)$ follow a bivariate Gaussian distribution, i.e.,
\begin{equation*}
    (X,Y)\sim\calN\left(\left[\begin{array}{cc}
        \mu_X  \\
        \mu_Y 
    \end{array}\right],\left[\begin{array}{cc}
        \sigma_X^2 & \rho\sigma_X\sigma_Y \\
        \rho\sigma_Y\sigma_X & \sigma_Y^2
    \end{array}\right]\right),
\end{equation*}
with $\mu_X,\mu_Y\in\bbR$, $\sigma_X,\sigma_Y\in(0,\infty)$, and $\rho\in[-1,1]$. Let $t_X$ and $t_Y$ denote the decision and target thresholds, respectively.

\subsection{Conditional event probability}\label{appendix:bivariate-gaussian-cep}

The conditional event probability is defined as $\mathrm{CEP}(x):=\Pr(Y>t_Y\mid X=x)$. The conditional distribution of $Y$ given $X=x$ is Gaussian and expressed as
\begin{equation*}
    Y\mid X=x\sim\calN(\mu_{Y\mid x}, \sigma_{Y\mid x}^2),
\end{equation*}
with $\mu_{Y\mid x}=\mu_Y+\rho\sigma_Y(x-\mu_X)/\sigma_X$ and $\sigma_{Y\mid x}^2=(1-\rho^2)\sigma_Y^2$. Hence,
\begin{equation*}
    \mathrm{CEP}(x)=1-\Phi\left(\frac{(t_Y-\mu_Y)/\sigma_Y-\rho(x-\mu_X)/\sigma_X}{\sqrt{1-\rho^2}}\right).    
\end{equation*}

\subsection{False alarm rate, hit rate, recall, precision}\label{appendix:bivariate-gaussian-far}

Let $\Phi(\cdot)$ and $\Phi(\cdot,\cdot;\rho)$ be the univariate and bivariate Gaussian cumulative distribution functions, respectively. 
\begin{align*}
    \FAR(t_X)&:=\Pr(X>t_X\mid Y\leq t_Y)\\
    &=1-\frac{\Pr(X\leq t_X, Y\leq t_Y)}{\Pr(Y\leq t_Y)}\\
    &=1-\frac{\Phi\left(\frac{t_X-\mu_X}{\sigma_X},\frac{t_Y-\mu_Y}{\sigma_Y};\rho\right)}{\Phi\left(\frac{t_Y-\mu_Y}{\sigma_Y}\right)}
\end{align*}

\begin{align*}
    \HR(t_X):=\Rec(t_X)&:=\Pr(X>t_X\mid Y> t_Y)\\
    &=\frac{1-\Pr(X\leq t_X)-\Pr(Y\leq t_Y)+\Pr(X\leq t_X, Y\leq t_Y)}{1-\Pr(Y\leq t_Y)}\\
    &=\frac{1-\Phi\left(\frac{t_X-\mu_X}{\sigma_X}\right)-\Phi\left(\frac{t_Y-\mu_Y}{\sigma_Y}\right)+\Phi\left(\frac{t_X-\mu_X}{\sigma_X},\frac{t_Y-\mu_Y}{\sigma_Y};\rho\right)}{1-\Phi\left(\frac{t_Y-\mu_Y}{\sigma_Y}\right)}
\end{align*}

\begin{align*}
    \Pre(t_X)&:=\Pr(Y> t_Y\mid X>t_X)\\
    &=\frac{1-\Pr(X\leq t_X)-\Pr(Y\leq t_Y)+\Pr(X\leq t_X, Y\leq t_Y)}{1-\Pr(X\leq t_X)}\\
    &=\frac{1-\Phi\left(\frac{t_X-\mu_X}{\sigma_X}\right)-\Phi\left(\frac{t_Y-\mu_Y}{\sigma_Y}\right)+\Phi\left(\frac{t_X-\mu_X}{\sigma_X},\frac{t_Y-\mu_Y}{\sigma_Y};\rho\right)}{1-\Phi\left(\frac{t_X-\mu_X}{\sigma_X}\right)}
\end{align*}

\section{Spatial distribution of event counts}

\begin{figure}[H]
    \centering
    \includegraphics[width=.85\linewidth]{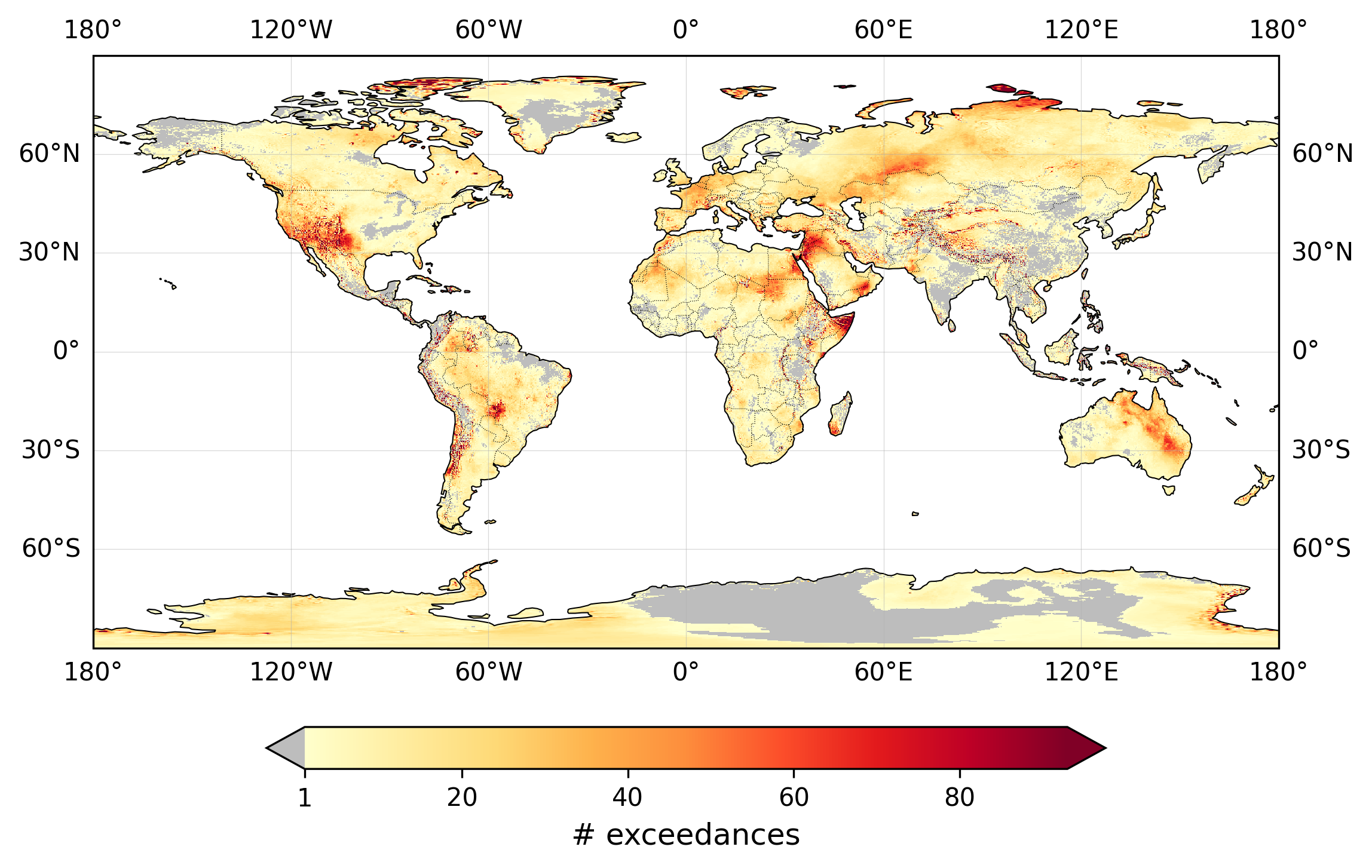}
    \caption{Spatial distribution of observed exceedances of the 98th climatological percentiles as described in Section~\ref{subsubsec:extreme-temperatures}.}
    \label{fig:counts-q98}
\end{figure}

\begin{figure}[H]
    \centering
    \includegraphics[width=.85\linewidth]{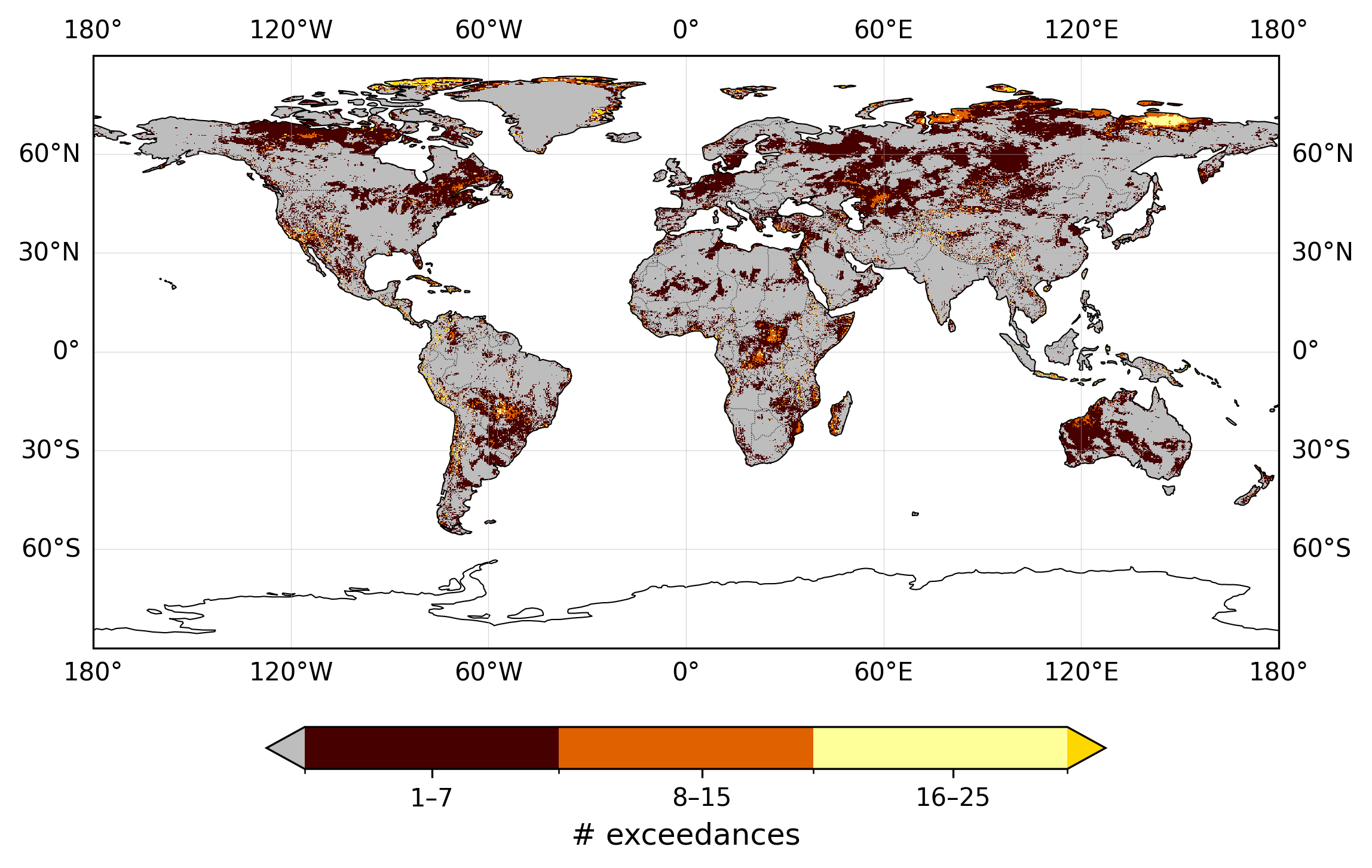}
    \caption{Spatial distribution of exceedances of observed record-breaking events as described in Section~\ref{subsubsec:extreme-temperatures}.}
    \label{fig:counts-max}
\end{figure}

\end{document}